\documentclass{article}
\usepackage[latin1]{inputenc}
\usepackage[height=9in,width=6.0in]{geometry}

\usepackage{amsmath,amsthm,amssymb,bm}
\usepackage[sort&compress]{natbib}
\usepackage{graphicx,mathrsfs}
\usepackage{paralist}
\usepackage{enumerate}
\usepackage{stmaryrd}
\usepackage{fancyhdr}
\usepackage{colortbl}

\newcommand{\mx}[1]{{\mathsf{#1}}}
\newcommand{\bmx}[1]{\overline{\mx{#1}}}
\newcommand{\tp}[2]{\left({#1}\right) \cdot {#2}}

\newcommand{\p}{\bm{P}}
\newcommand{\q}{\bm{Q}}
\newcommand{\s}{\bm{S}}
\newcommand{\rc}{\bm{R}}
\newcommand{\id}{\bm{I}}
\newcommand{\f}{\bm{F}}
\newcommand{\rr}{\bm{R}}

\newcommand{\sml}[1]{{\small #1}}

\newcommand{\frob}[1]{\| {#1} \|_\text{F}}

\newcommand{\C}{\mathcal{C}}
\newcommand{\mean}{\bm{\mu}}

\newcommand{\half}{\tfrac{1}{2}}

\newcommand{\x}{\bm{x}}

\newcommand{\ip}[2]{\left \langle {#1}, {#2} \right\rangle}

\newcommand{\set}[1]{\left\{ {#1} \right\}}

\newcommand{\breg}{B_f}

\newcommand{\jopt}{J_{\text{OPT}}}

\newcommand{\reals}{\mathbb{R}}

\newcommand{\E}{E}

\newcommand{\cc}[1]{\multicolumn{1}{c}{{#1}}}
\newcommand{\y}{\bm{y}}
\newcommand{\dist}{d}
\newcommand{\cotec}{CoTeC}

\definecolor{lightgray}{gray}{.80}

\newcommand{\gc}[1]{\resizebox{35pt}{5.5pt}{\raisebox{-1pt}{\colorbox{lightgray}{{#1}}}}}
\newcommand{\hc}[2]{\gc{${#1} \cdot 10^{#2}$}}

\DeclareMathOperator*{\argmin}{argmin}
\DeclareMathOperator*{\argmax}{argmax}

\numberwithin{equation}{section}
\bibpunct{[}{]}{,}{n}{,}{,}

\newtheorem{theorem}{Theorem}[section]
\newtheorem{lemma}[theorem]{Lemma}

\newtheorem{proposition}[theorem]{Proposition}

\theoremstyle{definition}
\newtheorem{example}[theorem]{Example}

\numberwithin{equation}{section}
\bibpunct{[}{]}{,}{n}{,}{,}

\title{Approximation Algorithms for Bregman Co-clustering and Tensor
  Clustering\thanks{A part of the theory of this paper appeared in  MPI-TR-\#177~\citep{kmeanstr}.}}

\author{
Stefanie Jegelka\\
MPI for Biological Cybernetics\\
72070 Tübingen, Germany
\and
Suvrit Sra\\
MPI for Biological Cybernetics\\
72070 Tübingen, Germany
\and
Arindam Banerjee\\
Univ. of Minnesota\\
MN 55455, USA
}

\date{}
 
\begin{document} 
\pagestyle{fancy}
\maketitle 
\hrule

\begin{abstract} 
  In the past few years powerful generalizations to the Euclidean k-means
  problem have been made, such as Bregman clustering~\citep{banerjee},
  co-clustering (i.e., simultaneous clustering of rows and columns of an input
  matrix)~\citep{itcc,coclust}, and tensor
  clustering~\citep{arinmulti,shashua}. Like k-means, these more general
  problems also suffer from the NP-hardness of the associated
  optimization. Researchers have developed approximation algorithms of varying
  degrees of sophistication for k-means, k-medians, and more recently also for
  Bregman clustering~\citep{ackerman}. However, there seem to be no
  approximation algorithms for Bregman co- and tensor clustering. In this
  paper we derive the first (to our knowledge) guaranteed methods for these
  increasingly important clustering settings. Going beyond Bregman
  divergences, we also prove an approximation factor for tensor clustering
  with arbitrary separable metrics. Through extensive experiments we
  evaluate the characteristics of our method, and show that it also has practical impact.
\end{abstract} 

\section{Introduction}
Partitioning data points into clusters is a fundamentally hard problem. The
well-known Euclidean k-means problem that partitions the input data points
(vectors in $\reals^d$) into $K$ clusters while minimizing sums of their
squared distances to corresponding cluster centroids, is an NP hard
problem~\citep{drineas} (exponential in $d$). However, simple and frequently
used procedures that rapidly obtain local minima exist since a long
time~\citep{lloyd,hartigan}.

Because of its wide applicability and importance, the Euclidean k-means
problem has been generalized in several directions. Specific
examples relevant to this paper include:
\begin{itemize}
  \setlength{\itemsep}{-1pt}
\item \emph{Bregman clustering}~\citep{banerjee}, where instead of minimizing squared
  Euclidean distances one minimizes Bregman divergences (which are generalized
  distance functions, see~(\ref{eq:4}) or~\citep{censor:zenios} for details),
\item \emph{Bregman co-clustering}~\citep{coclust} (which includes both
  Euclidean~\citep{suv.coclus} and information-theoretic
  co-clustering~\citep{itcc} as special cases), where the set of input
  vectors is viewed as a matrix and one \emph{simultaneously} clusters rows
  and columns to obtain coherent submatrices (co-clusters), while minimizing
  a Bregman divergence, and
\item \emph{Tensor clustering} or multiway clustering~\citep{shashua},
  especially the version based on Bregman divergences~\citep{arinmulti}, where
  one simultaneously clusters along various dimensions of the input tensor.
\end{itemize}
For these problems too, the commonly used heuristics perform well, but do not
provide theoretical guarantees (or at best assure local optimality). For
k-means type clustering problems---i.e., problems that group together input
vectors into clusters while minimizing ``distance'' to cluster
centroids---there exist several algorithms that approximate a globally optimal
solution. We refer the reader to~\citep{dasv07,kumar,ackerman,acker2}, and the
numerous references therein for more details.

In stark contrast, approximation algorithms for tensor clustering are much less
studied. We are aware of only two very recent attempts (both papers are from
2008) for the two-dimensional special case of co-clustering,
namely,~\citep{coclus} and \citep{coclus1}---and both of the papers 
follow similar approaches to obtain their approximation guarantees. Both prove
a $2\alpha_1$-approximation for Euclidean co-clustering, \citet{coclus1} an
additional factor of $(1 + \sqrt{2})$ for binary matrices and an $\ell_1$ norm
objective, and \citet{coclus} a factor of $3\alpha_1$ for co-clustering real
matrices with $\ell_p$ norms. In all factors $\alpha_1$ is an approximation guarantee
for clustering either rows or columns.  In this paper, we build
upon~\citep{coclus} and obtain approximation algorithms for tensor
clustering with Bregman divergences and arbitrary separable metrics such as
$\ell_p$-norms.  The latter result is of
particular interest for $\ell_1$-norm based tensor clustering, which may be
viewed as a generalization of k-medians to tensors.
In the terminology
of~\citep{banerjee}, we focus on the ``block average'' versions of co- and
tensor clustering.

Additional discussion and relevant references for co-clustering can be found
in~\citep{coclust}, while for the lesser known problem of tensor clustering
more background can be gained by referring
to~\citep{tensor1,tensor2,tensor3,tensor4,arinmulti,shashua}.

\subsection{Contributions}
\vspace{-5pt}%
The main contribution of this paper is the analysis of an approximation
algorithm for tensor clustering that achieves an approximation ratio of
$O(m\alpha)$, where $m$ is the order of the tensor
and $\alpha$ is the approximation factor of a
corresponding 1D clustering algorithm. Our results apply to a fairly broad
class of objective functions, including metrics such as $\ell_p$ norms or
Hilbertian metrics~\cite{scsm01,hebo05}, and divergence functions such as
Bregman divergences~\cite{censor:zenios} (with some assumptions). As
corollaries, our results solve two open problems posed by~\cite{coclus},~viz.,
whether their methods for Euclidean co-clustering could be extended to Bregman
co-clustering, and if one could extend the approximation guarantees to tensor
clustering. Owing to the structure of the algorithm, our results also
give insight into proprties of the tensor clustering problem as such,
namely, a bound on the amount of information inherent in the
joint consideration of several dimensions.

In addition, we provide extensive experimental validation of the
theoretical claims, which forms an additional contribution of this
paper.

\section{Background}
Traditionally, ``center'' based clustering algorithms seek partitions of
columns of an input matrix $\mx{X} = [\x_1,\ldots,\x_n]$ into clusters $\C =
\set{\C_1,\ldots,\C_K}$, and find ``centers'' $\mean_k$ that minimize the objective
\begin{equation}
  \label{eq:5}
  J(\C) = \sum\nolimits_{k=1}^K\sum\nolimits_{\x \in \C_k} d(\x, \mean_k),
\end{equation}
where the function $d(\x,\y)$ measures cluster quality. The ``center''
$\mean_k$ of cluster $\C_k$ is given by the mean of the points in $\C_k$ when
$d(\x,\y)$ is a Bregman divergence~\cite{banerjee}.  Co-clustering
extends~\eqref{eq:5} to seek simultaneous partitions (and centers
$\mean_{IJ}$) of rows and columns of $\mx{X}$, so that the objective function
\begin{equation}
  \label{eq:11}
  J(\C) = \sum\nolimits_{I,J} \sum\nolimits_{i \in I, j \in J} d(x_{ij}, \mean_{IJ}),
\end{equation}
is minimized; $\mean_{IJ}$ denotes the (scalar) ``center'' of the cluster
described by the row and column index sets, viz., $I$ and
$J$. Formulation~\eqref{eq:11} is easily generalized to tensors, as shown in
Section~\ref{sec:probl} below. However, we first recall basic notation about
tensors before formally presenting the tensor clustering problem.
Tensors are well-studied in multilinear algebra~\citep{greub}, and
they are gaining importance in both data mining and machine learning
applications~\citep{kddwkshp,jimeng}.

\subsection{Tensors}
A large part of the material in this section is derived from the well-written
paper of~\citet{leklim}---their notation turns out to be particularly suitable
for our analysis. An order-$m$ tensor $\mx{A}$ may be viewed as an element of
the vector space $\reals^{{n_1} × \ldots × n_m}$ (in this paper we denote matrices
and tensors using sans-serif letters). An individual component of the tensor
$\mx A$ is represented by the multiply-indexed value $a_{i_1i_2\ldots i_m}$, where
$i_j \in \set{1,\dots,n_j}$ for $1 \leq j \leq m$.

\subsubsection*{Multilinear matrix multiplication} For us the most important
operation on tensors is that of multilinear matrix multiplication, which is
a generalization of the familiar concept of matrix multiplication. Matrices
\emph{act} on other matrices by either left or right multiplication. For an
order-3 tensor, there are three dimensions on which a matrix may act via
matrix multiplication. For example, for an order-3 tensor $\mx{A} \in
\reals^{n_1 × n_2 × n_3}$, and three matrices $\mx{P} \in
\reals^{p_1 × n_1}$, $\mx{Q} \in \reals^{p_2× n_2}$, and $\mx{R}
\in \reals^{p_3 × n_3}$, \emph{multilinear matrix multiplication} is
the operation defined by the action of these three matrices on the different
dimensions of $\mx{A}$ that yields the tensor $\mx{A}' \in \reals^{p_1
  × p_2 × p_3}$. Formally, the entries of the tensor $\mx{A}'$ are
given by
\begin{equation}
  \label{eq:1}
  a_{lmn}' = \sum\nolimits_{i,j,k=1}^{n_1,n_2,n_3} p_{li}q_{mj}r_{nk}a_{ijk},
\end{equation}
and this operation is written compactly as
\begin{equation}
  \label{eq:21}
  \mx{A}' = \bigl(\mx{P},\mx{Q},\mx{R}\bigr) \cdot \mx{A}.
\end{equation}

Multilinear multiplication extends naturally to tensors of arbitrary
order. If $\mx{A} \in \reals^{n_1× n_2×\cdots×
  n_m}$, and $\mx{P}_1 \in \reals^{p_1 × n_1}, \ldots, \mx{P}_m \in
\reals^{p_m × n_m}$, then $\mx{A}' =
( (\mx{P}_1,\ldots,\mx{P}_m) \cdot \mx{A}) \in \reals^{p_1 × \cdots ×  p_m} $ has components
\begin{equation}
  \label{eq:22}
  a_{i_1i_2\ldots i_m}' = \sum\nolimits_{j_1,\ldots,j_m=1}^{n_1,\ldots,n_m}
  p_{i_1j_1}^{(1)}\cdots p_{i_mj_m}^{(m)}a_{j_1\ldots j_m}, 
\end{equation}
where $p_{ij}^{(k)}$ denotes the $ij$-th entry of matrix $\mx{P}_k$.

\begin{example}[Matrix Multiplication] Let $\mx{A} \in \reals^{n_1 ×
    n_2}$, $\mx{P} \in \reals^{p × n_1}$, and $\mx{Q} \in \reals^{q
    × n_2}$ be three matrices. The matrix product $\mx{PAQ}^\top$ can be written as
  the multilinear multiplication $\tp{\mx{P},\mx{Q}}{\mx{A}}$.
\end{example}

\begin{proposition}[Basic Properties]
  The following properties of multilinear multiplication are easily verified
  (and generalized to tensors of arbitrary order):
  \begin{enumerate}
  \item \textbf{Linearity:} Let $\alpha$, $\beta \in \reals$, and $\mx{A}$
    and $\mx{B}$ be tensors with same dimensions, then
    \begin{equation*}
      \tp{\mx{P},\mx{Q}}{(\alpha\mx{A}+ \beta\mx{B})} =
      \alpha\tp{\mx{P},\mx{Q}}{\mx{A}} + \beta \tp{\mx{P},\mx{Q}}{\mx{B}}
    \end{equation*}
  \item \textbf{Product rule:} For matrices $\mx{P}_1$, $\mx{P}_2$,
    $\mx{Q}_1$, $\mx{Q}_2$ of appropriate dimensions, and a tensor $\mx{A}$
    \begin{equation*}
      \tp{\mx{P}_1,\mx{P}_2}{\bigl(\tp{\mx{Q}_1,\mx{Q}_2}{\mx{A}}\bigr)} = 
      \tp{\mx{P}_1\mx{Q}_1, \mx{P}_2\mx{Q}_2}{\mx{A}}
    \end{equation*}
  \item \textbf{Multilinearity:} Let $\alpha$, $\beta \in \reals$, and
    $\mx{P}$, $\mx{Q}$, and $\mx{R}$ be matrices of appropriate
    dimensions. Then, for a tensor $\mx{A}$ the following holds
    \begin{equation*}
      \tp{\mx{P}, \alpha\mx{Q}+\beta\mx{R}}{\mx{A}} =
      \alpha\tp{\mx{P},\mx{Q}}{\mx{A}} + \beta\tp{\mx{P},\mx{R}}{\mx{A}}
    \end{equation*}
  \end{enumerate}
\end{proposition}

\subsubsection*{Vector Norms} The standard vector $\ell_p$-norms can be easily
extended to tensors, and are defined as
\begin{equation}
  \label{eq:3}
  \|\mx{A}\|_p = \Bigl(\sum\nolimits_{i_1,\ldots,i_m} |a_{i_1\ldots i_m}|^p\Bigr)^{1/p},
\end{equation}
for $p \geq 1$. In particular for $p=2$ we get the ``Frobenius'' norm, also
written as $\frob{\mx{A}}$.

\subsubsection*{Inner Product} The Frobenius norm induces an inner-product
that can be defined as
\begin{equation}
  \label{eq:30}
  \ip{\mx{A}}{\mx{B}} = \sum\nolimits_{i_1,\ldots,i_m} a_{i_1\ldots i_m}b_{i_1\ldots i_m},
\end{equation}
so that $\frob{\mx{A}}^2 = \ip{\mx{A}}{\mx{A}}$ holds as usual. 
\begin{proposition}
  The following property of this inner product is easily verified (a
  generalization of the familiar property $\ip{\mx{Ax}}{\mx{By}} =
  \ip{\mx{x}}{\mx{A}^\top\mx{By}}$ for vectors):
\begin{equation}
  \label{eq:31}
  \begin{split}
    \ip{\tp{\mx{P}_1,\ldots,\mx{P}_m}{\mx{A}}}{\tp{\mx{Q}_1,\ldots,\mx{Q}_m}{\mx{B}}}  = \ip{\mx{A}}{\tp{\mx{P}_1^\top\mx{Q}_1,\ldots,\mx{P}_m^\top\mx{Q}_m}{\mx{B}}}.
  \end{split}
\end{equation}
\end{proposition}
\emph{Proof:} Using definition~(\ref{eq:22}) and the inner-product
rule~\eqref{eq:30} we have {\small
\begin{equation*}
  \begin{split}
    &\ip{\tp{\mx{P}_1,\ldots,\mx{P}_m}{\mx{A}}}{\tp{\mx{Q}_1,\ldots,\mx{Q}_m}{\mx{B}}}
    = \sum_{i_1,\ldots,i_m}\mathop{\sum_{j_1,\ldots,j_m}}_{k_1,\ldots,k_m}
    p_{i_1j_1}^{(1)}q_{i_1k_1}^{(1)}\cdots
    p_{i_mj_m}^{(m)}q_{i_mk_m}^{(m)}a_{j_1\ldots j_m}b_{k_1\ldots k_m},\\
    &= \mathop{\sum_{j_1,\ldots,j_m}}_{k_1,\ldots,k_m}
    \bigl(\sum_{i_1}p_{i_1j_1}^{(1)}q_{i_1k_1}^{(1)}\bigr)\cdots
    \bigl(\sum_{i_m}p_{i_mj_m}^{(m)}q_{i_mk_m}^{(m)}\bigr)a_{j_1\ldots
      j_m}b_{k_1\ldots k_m}\\
    &= \mathop{\sum_{j_1,\ldots,j_m}}_{k_1,\ldots,k_m}
    (\mx{P}_1^\top\mx{Q}_1)_{j_1k_1}\cdots
    (\mx{P}_m^\top\mx{Q}_m)_{j_mk_m}a_{j_1\ldots j_m}b_{k_1\ldots k_m}  =
    \sum_{j_1\ldots j_m} a_{j_1\ldots j_m} b'_{j_1\ldots j_m}  =
    \ip{\mx{A}}{\mx{B}'}, 
  \end{split}
\end{equation*}}
where $\mx{B}' =
\tp{\mx{P}_1^\top\mx{Q}_1,\ldots,\mx{P}_m^\top\mx{Q}_m}{\mx{B}}$.

\subsubsection*{Divergence}
Finally, we define an arbitrary \emph{divergence} function $\dist(\mx{X},\mx{Y})$
between two $m$-dimensional tensors $X,Y$ as an elementwise sum of individual divergences, i.e.,
\begin{equation}
  \label{eq:13}
  \dist(\mx{X}, \mx{Y}) = \sum\nolimits_{i_1,\ldots,i_m}
  \dist(x_{i_1,\ldots,i_m},y_{i_1,\ldots,i_m}),
  \vspace*{-5pt}
\end{equation}
and we will define the scalar divergence $d(x,y)$ as the need arises.

\subsection{Tensor clustering}
\label{sec:probl}
Let $\mx{A} \in \reals^{n_1 \times \cdots \times n_m}$ be an order-$m$ tensor that we wish to
partition into coherent sub-tensors (or clusters). A basic approach is to
minimize the sum of the divergences between individual (scalar) elements in
each cluster to their corresponding (scalar) cluster ``centers''. Readers
familiar with~\cite{coclust} will recognize this to be a ``block-average''
variant of tensor clustering.

Assume that each dimension $j$ (\sml{$1\leq j \leq m$}) is partitioned into $k_j$
clusters. Let $\mx{C}_j \in \{0,1\}^{n_j \times k_j}$ be the cluster indicator matrix
for dimension $j$, where the $ik$-th entry of such a matrix is one if and only
if index $i$ belongs to the $k$-th cluster (\sml{$1\leq k\leq k_j$}) for dimension
$j$. Then, the \emph{tensor clustering} problem is (\emph{cf.}~\ref{eq:11}):
\begin{equation}
  \label{eq:2}
    \underset{\mx{C}_1,\ldots,\mx{C}_m, \mx{M}}{\text{minimize}}\quad
    \dist(\mx{A},\; (\mx{C}_1, \ldots, \mx{C}_m) \cdot
    \mx{M}),\quad
    \text{s.t. } \mx{C}_j \in \set{0,1}^{n_j\times k_j},
\end{equation}
where the tensor $\mx{M}$ collects all the cluster ``centers.''

\section{Algorithm and Analysis}
\label{sec:coclust}
Given formulation~\eqref{eq:2}, our algorithm, which we name
\textbf{Co}mbination \textbf{Te}nsor \textbf{C}lustering (CoTeC), follows the
simple outline:
\vspace*{3pt}

\fbox{
  \begin{minipage}[c]{0.9\linewidth}
    \vspace*{5pt}
    \flushleft
    \begin{enumerate}
    \item Cluster along each dimension $j$, using an approximation algorithm
      to obtain clustering $\mx{C}_j$; Let $\bm{C} = (\mx{C}_1,\ldots,\mx{C}_m)$
    \item Compute $\mx{M} = \argmin_{\mx{X} \in \mathbb{R}^{k_1 \times \cdots \times k_m}} d(\mx{A}, \bm{C} \cdot \mx{X})$.
    \item Return the tensor clustering $(\mx{C}_1,\ldots,\mx{C}_m)$  (with representatives $\mx{M}$).
  \end{enumerate}
  \end{minipage}
}\\
Instead of clustering one dimension at a time, we can also cluster along $t$
dimensions at a time, which we will call $t$-\emph{dimensional clustering}.
For an order-$m$ tensor, with $t=1$ we form groups of order-$(m-1)$
tensors. For illustration, consider an order-$3$ tensor $\mx{A}$ for which we
group matrices when $t=1$. For the first dimension we cluster the objects
$\mx{A}(i,:,:)$ (using \textsc{Matlab} notation) to obtain cluster indicators
$\mx{C}_1$; we repeat the procedure for the second and third dimensions. The
approximate tensor clustering will be the combination $(\mx{C}_1,
\mx{C}_2,\mx{C}_3)$. As we assumed the Bregman divergences to be separable,
the sub-tensors, e.g., $\mx{A}(i,:,:)$ can be simply treated as vectors.

Apart from~\citep{coclus,coclus1}, all approximation guarantees refer to
one-dimensional clustering algorithms. Any one-dimensional approximation
algorithm can be used as a base method for our scheme outlined above. For
example, the method of~\citet{acker2}, or the more practical Bregman
clustering approaches
of~\citep{ecmlbreg,kmeanstr}\footnote{Both~\citep{ecmlbreg,kmeanstr}
  discovered essentially the same method for Bregman clustering, though the
  analysis of~\citep{ecmlbreg} is somewhat sharper.} are two potential
choices, though with different approximation factors. Clustering along
individual dimensions and then combining the results to obtain a tensor
clustering might seem counterintuitive to the idea of
``co''-clustering, where one \emph{simultaneously}
clusters along different dimensions.. However, our analysis will
show that
dimension-wise clustering suffices to obtain strong approximation guarantees
for tensor clustering---a fact often observed empirically
too. At the same time, our Thmeorem 1 bounds the amount of
information that can lie in the simultaneous consideration of multiple dimensions.

\subsection{Results}
The main contribution of this paper is the following approximation guarantee
for \cotec, which we prove in the remainder of this section.

\vspace{-5pt}
\begin{theorem}[Approximation]
  \label{thm:mway}
  Let $\mx{A}$ be an order-$m$ tensor and let $\C_j$ denote its clustering
  along the $j$th subset of $t$ dimensions ($1 \leq j \leq m/t$), as obtained from a
  multiway clustering algorithm with guarantee $\alpha_t$\footnote{We say an
    approximation algorithm has guarantee $\alpha$ if it yields a solution that
    achieves an objective value within a factor $O(\alpha)$ of the optimum.}. Let
  $\C = (\C_1,\ldots,\C_{m/t})$ denote the induced tensor clustering,
  and $\jopt(m)$ the best $m$-dimensional clustering. Then,
  \begin{equation}
    \label{eq:18}
    J(\C) \leq p(m/t) \rho_\dist \alpha_t\jopt(m), \quad \text{ with }
  \end{equation}
  \begin{enumerate}
  \item $\rho_\dist = 1$ and $p(m/t) = 2^{\log_2 m/t}$ if $\dist(x,y) = (x-y)^2$,
  \item $\rho_\dist = 1$ and $p(m/t) = 2 m/t$ if $\dist(x,y)$ is a metric\footnote{The results can be trivially
      extended to $\lambda$-relaxed metrics that satisfy $d(x,y) \leq \lambda(d(x,z) +
      d(z,y))$; the corresponding approximation factor just gets scaled by
      $\lambda$.}.
  \end{enumerate}
\end{theorem}

Theorem~\ref{thm:mway} is quite general, and it can be combined with some
natural assumptions (see Section \ref{sec:impl}) to yield results for tensor
clustering with more general divergence functions too (here $\rho_d$ might be greater than
1).

\subsection{Analysis:  Theorem~\ref{thm:mway}, Euclidean case}
\label{sec:analysis-actual}
We begin our proof with the Euclidean case, i.e., $\dist(x,y) = (x-y)^2$.  Our
proof is inspired by the techniques of~\cite{coclus}. 
We establish that given a clustering algorithm which clusters along $t$ of the
$m$ dimensions at a time\footnote{One could also consider clustering
  differently sized subsets of the dimensions, say $\set{t_1,\ldots,t_r}$, where
  $t_1+\cdots+t_r = m$. However, this requires unilluminating notational jugglery,
  which we can skip for simplicity of exposition.} with an approximation
factor of $\alpha_t$, our CoTeC algorithm achieves an objective within a
factor $O(\lceil m/t\rceil\alpha_t)$ of the optimal. For example, for $t=1$ we can use the
seeding methods of \cite{ecmlbreg,kmeanstr} or the stronger approximation
algorithms of~\cite{acker2}. We assume without loss of generality (wlog) that
$m = 2^ht$ for an integer $h$ (otherwise, pad in empty dimensions).

Since for the squared Frobenius norm, each cluster ``center'' is given by the
mean, we can recast Problem~\eqref{eq:2} into a more convenient form. To that
end, note that the individual entries of the means tensor $\mx{M}$ are given
by (\emph{cf.}~(\ref{eq:11}))
\begin{equation}
  \label{eq:28}
  \mean_{I_1\ldots I_m} = \frac{1}{|I_1|\cdots |I_m|} \sum_{i_1 \in I_1,\ldots, i_m \in I_m} a_{i_1\ldots i_m},
\end{equation}
with index sets $I_j$ for \sml{$1 \leq j \leq m$}. Let
$\bmx{C}_j$ be the normalized cluster indicator matrix obtained by normalizing
the columns of $\mx{C}_j$, so that $\bmx{C}_j^\top\bmx{C}^{}_j =
\mx{I}_{k_j}$. Then, we can rewrite~\eqref{eq:2} in terms of projection
matrices $\mx{P}_j$ as:
\begin{equation}
  \label{eq:29}
    \underset{\C = (\bmx{C}_1,\ldots,\bmx{C}_m)}{\text{minimize}}\quad J(\C) =
    \frob{\mx{A} - (\mx{P}_1, \ldots, \mx{P}_m)\cdot \mx{A}}^2,\quad
    \text{s.t.}\ \mx{P}_j = \bmx{C}^{}_j\bmx{C}_j^\top.
\end{equation}

\begin{lemma}[Pythagorean]
  \label{lemm:trace}
  Let $\p = (\mx{P}_1, \ldots, \mx{P}_t)$, $\s =
  (\mx{P}_{t+1},\ldots,\mx{P}_m)$, and $\p^\bot  = (\mx{I} - \mx{P}_1,
  \ldots,   \mx{I} - \mx{P}_t)$ be \emph{collections} of projection matrices
  $\mx{P}_j$. Then,
  {\small
    \begin{equation*}
      \begin{split}
        \|\tp{\p, \s}{\mx{A}} + ({\p^\bot,\rc}) \cdot \mx{B}\|^2 =
        \|({\p,\s}) \cdot {\mx{A}}\|^2 + \|({\p^\bot,\rc}) \cdot {\mx{B}}\|^2, 
      \end{split}
    \end{equation*}}%
  where $\rc$ is a collection of $m-t$ projection matrices.
\end{lemma}
\begin{proof}  
  Using $\frob{\mx{A}}^2 = \ip{\mx{A}}{\mx{A}}$ we can rewrite the l.h.s. as
\begin{equation*}
  \begin{split}
    &\|\tp{\p, \s}{\mx{A}} + \tp{\p^\bot,\rc}{\mx{B}}\|^2\\
    & = \|\tp{\p,\s}{\mx{A}}\|^2 + \|\tp{\p^\bot,\rc}{\mx{B}}\|^2 +
    2\bigl\langle\tp{\p, \s}{\mx{A}}, \tp{\p^\bot,\rc}{\mx{B}}\bigr\rangle. 
  \end{split}
\end{equation*}
The last term is immediately seen to be zero using Property~(\ref{eq:31}) and
the fact that $\mx{P}_j^\top\mx{P}_j^\bot = \mx{P}_j(\mx{I} - \mx{P}_j) =
\mathbf{0}$.
\end{proof}

\textbf{Some more notation:} Since we cluster along $t$ dimensions at a time,
we recursively partition the initial set of all $m$ dimensions until (after
$\log(m/t)+1$ steps), the sets of dimensions have length $t$. Let $l$ denote
the level of recursion, starting at $l=\log(m/t)=h$ and going down to
$l=0$. At level $l$, the sets of dimensions will have length $2^lt$ (so that
for $l=0$ we have $t$ dimensions). We represent each clustering along a subset
of $2^lt$ dimensions by its corresponding $2^lt$ projection matrices. We
gather these projection matrices into the collection $\p^l_i$ (note boldface),
where the index $i$ ranges from $1$ to $2^{h-l}$.

\begin{example}
  \label{ex:clust.part}
  Consider an order-$8$ tensor where we group $t=2$ dimensions at a
  time. Then, $h=\log(m/t)=2$ and we have $3$ levels. We recursively divide the set
  of dimensions in the middle, i.e., $\{1,\ldots,8\}$ into $\{1,\ldots, 4\}$
  and $\{5, \ldots, 8\}$ and so on, ending with $\{\{1,2\}, \{3,4\},
  \{5,6\}, \{7,8\}\}$. The projection matrix for dimension $i$ is
  $\mx{P}_i$, and the full tensor clustering is represented by $(\mx{P}_1,
  \ldots, \mx{P}_8)$. For each level $l=0,1,2$, individual collections of
  projection matrices $\p_i^l$ are
\begin{equation*}
  \begin{split}
    &\p_1^2 = (\mx{P}_1, \, \mx{P}_2, \, \mx{P}_3, \,  \mx{P}_4, \, \mx{P}_5, \, \mx{P}_6, \, \mx{P}_7, \, \mx{P}_8)\\
    &\p_1^1 = (\mx{P}_1, \, \mx{P}_2, \, \mx{P}_3, \,  \mx{P}_4),\quad
    \p_2^1 = (\mx{P}_5, \, \mx{P}_6, \, \mx{P}_7, \, \mx{P}_8)\\
    &\p_1^0 = (\mx{P}_1, \, \mx{P}_2),\quad \ldots, %
    \quad\p^0_4 = (\mx{P}_7, \, \mx{P}_8).
  \end{split}
\end{equation*}
\end{example}

We also need some notation to represent a complete tensor clustering along all
$m$ dimensions, where \emph{only a subset} of $2^lt$ dimensions are
clustered. We pad the collection $\p^l_i$ with $m - 2^lt$ identity
matrices for the non-clustered dimensions, and call this padded collection
$\q^l_i$.
With recursive partitioning of the dimensions, $\q_i^l$ subsumes $\q_j^0$ for
\sml{$2^l(i-1) < j \leq 2^li$}, so that {\small
  \begin{equation*}
    \q_i^l= \prod\nolimits_{j=2^l(i-1)+1}^{2^li} \q_j^0.
  \end{equation*}
}%
At level $0$, the algorithm yields the collections $\q_i^0$ and $\p_i^0$. The
remaining clusterings are simply \emph{combinations}, i.e., products of these
level-0 clusterings.  We denote the collection of $m-2^lt$ identity matrices
(of appropriate size) by $\id^{l}$, so that $\q_1^l = (\p_1^l, \id^l)$.
Accoutered with our notation, we now prove the main lemma that relates the
combined clustering to its sub-clusterings.

\begin{lemma}
  \label{lemm:rec}
  Let $\mx{A}$ be an order-$m$ tensor and $m \geq 2^lt$. The objective
  function for any $2^lt$-dimensional clustering $\p_i^l =
  (\p_{2^l(i-1)+1}^0, \ldots,\p_{2^li}^0)$ can be bound via the
  sub-clusterings along only one set of dimensions of size $t$ as
  \begin{equation}
    \label{eq:2normT}
    \begin{split}
      \frob{\mx{A} - \q_i^l \cdot \mx{A}}^2 %
      \leq \max_{2^l(i-1) < j \leq 2^li} 2^{l} \frob{\mx{A} - \q_j^{0}\cdot \mx{A}}^2.
    \end{split}
  \end{equation}
\end{lemma}
We can always (wlog) permute dimensions so that any set of $2^l$ clustered
dimensions maps to the first $2^l$ ones. Hence, it suffices to prove the lemma
for $i=1$, i.e., the first $2^l$ dimensions.

\begin{proof}
  We prove the lemma for $i=1$ by induction on $l$.

  \emph{Base:} Let $l=0$. Then $\q^l_1 = \q_1^0$, and~\eqref{eq:2normT}
  holds trivially.

  \emph{Induction:} Assume the claim holds for $l \geq 0$. Consider a
  clustering $\p_1^{l+1} = (\p_1^l,\p_2^l)$, or equivalently $\q_1^{l+1} =
  \q_1^l\q_2^l$.  Using $\bm{P}+\bm{P}^\bot = \bm{I}$, we decompose $\mx{A}$ as
  {\small
  \begin{align*}
    \mx{A} &\quad= (\p_1^{l+1}+ {\p_1^{l+1}}^\bot, \id^{l+1}) \cdot \mx{A}
    \quad=\quad (\p_1^l + {\p_1^l}^\bot, \p_2^l + {\p_2^l}^\bot, \id^{l+1}) \cdot
    \mx{A}\\
    &\quad= (\p_1^{l},\p_2^{l},\id^{l+1})\cdot \mx{A} + ({\p_1^{l}}^\bot, \p_2^{l},
    \id^{l+1})\cdot \mx{A} + (\p_1^{l}, {\p_2^{l}}^\bot, \id^{l+1})\cdot \mx{A} + (
    {\p_1^{l}}^\bot, {\p_2^{l}}^\bot, \id^{l+1}) \cdot \mx{A}\\
    &\quad= \q_1^l\q_2^l \cdot \mx{A} + {\q_1^l}^\bot \q_2^l \cdot \mx{A} + \q_1^l{\q_2^l}^\bot \cdot
    \mx{A} + {\q_1^l}^\bot {\q_2^l}^\bot \cdot \mx{A},
  \end{align*}}%
where \sml{${\q_1^l}^\bot = ({\p_1^{l}}^\bot, \id^{l})$}. Since \sml{$\q_1^{l+1} =
  \q_1^l \q_2^l$}, the Pythagorean Property~\ref{lemm:trace} yields
{\small
\begin{equation*}
  \|\mx{A} - \q_1^{l+1}\cdot \mx{A}\|^2 = \|{\q_1^l}^\bot \q_2^l\cdot \mx{A}\|^2 + \|\q_1^l
  {\q_2^l}^\bot \cdot \mx{A}\|^2 + \|{\q_1^l}^\bot  {\q_2^l}^\bot\cdot \mx{A}\|^2. 
\end{equation*}
}%
Combining the above equalities with the assumption (wlog)
{\small
  \begin{equation*}
    \|{\q_1^l}^\bot \q_2^l \cdot \mx{A}\|^2  \geq  \|\q_1^l  {\q_2^l}^\bot\cdot  \mx{A}\|^2, 
  \end{equation*}
}
we obtain the inequalities {\small
  \begin{align*}
    &\|\mx{A} - \q_1^l\q_2^l\cdot \mx{A}\|^2\ \ \leq\ \ 2 \bigl(\|{\q_1^l}^\bot \q_2^l\cdot \mx{A}\|^2 + \|{\q_1^l}^\bot {\q_2^l}^\bot \cdot \mx{A}\|^2\bigr)\\
    &= 2 \|{\q_1^l}^\bot \q_2^l\cdot \mx{A} + {\q_1^l}^\bot {\q_2^l}^\bot \cdot \mx{A}\|^2\  =\
    2\|{\q_1^l}^\bot(\q_2^l+ {\q_2^l}^\bot)\cdot \mx{A}\|^2\\
    &= 2\|{\q_1^l}^\bot\cdot \mx{A}\|^2 \ = \ 2 \|\mx{A} - \q_1^{l}\cdot \mx{A}\|^2\\
    &\leq  2\max_{1 \leq j  \leq 2^l} \|\mx{A} - \q_j^{l}\cdot \mx{A}\|^2\ \ \leq\ \ 2\cdot 2^{l} \max_{1 \leq j \leq 2^{l+1}} \|\mx{A} -  \q_j^{0}\cdot \mx{A}\|^2,
  \end{align*}}%
where the last step follows from the induction hypothesis~(\ref{eq:2normT}),
and the two norm terms in the first line are combined using the Pythagorean
Property.
\end{proof}

\begin{proof} (\emph{Thm.~\ref{thm:mway}, Case 1}). Let $m=2^ht$. Using an
  algorithm with guarantee $\alpha_t$, we cluster each subset (indexed by $i$) of
  $t$ dimensions to obtain $\q_i^0$. Let $\s_i$ be the optimal sub-clustering
  of subset $i$, i.e., the result that $\q_i^0$ would be if $\alpha_t$ were $1$. We
  bound the objective for the collection of all $m$ sub-clusterings $\p_1^h =
  \q_1^h$ as
  \begin{equation}
    \label{eq:14}
    \frob{\mx{A}-\q_1^{h}\cdot \mx{A}}^2 \leq 2^h \max_j \|\mx{A} -  \q_j^{0}\cdot
    \mx{A}\|^2 \leq  2^h \alpha_t \max_j \frob{\mx{A} - \s_j\cdot \mx{A}}^2.
  \end{equation}
  The first inequality follows from Lemma~\ref{lemm:rec}, while the last
  inequality follows from the $\alpha_t$ approximation factor that we
  used to get sub-clustering $\q_j^{0}$.

  So far we have related our approximation to an optimal sub-clustering
  across a set of dimensions.  Let us hence look at the relation between
  such an optimal sub-clustering $\s$ of the first $t$ dimensions (via
  permutation, these dimensions correspond to an arbitrary subset of size
  $t$), and the optimal tensor clustering $\f$ across all the $m=2^ht$
  dimensions. Recall that a clustering can be expressed by either the
  projection matrices collected in $\q^l_1$, or by cluster indicator
  matrices $\mx{C}_i$ together with the mean tensor $\mx{M}$, so that
  $$(\mx{C}_1, \ldots, \mx{C}_{2^lt}, \id^{l}) \cdot \mx{M} = \q_1^l \cdot\mx{A}.$$

  Let $\mx{C}^S_j$ and $\mx{C}^F_j$ be the dimension-wise cluster indicator
  matrices for $\s$ and $\f$, respectively; 
  By definition, $\s$ solves
  \begin{equation*}
    \begin{split}
      \underset{\mx{C}_1,\ldots,\mx{C}_t, \mx{M}}{\text{min}}\ \
      \frob{\mx{A} - (\mx{C}_1, \ldots, \mx{C}_t, \id^0) \cdot
        \mx{M}}^2,\quad\text{s.t. } \mx{C}_j \in \set{0,1}^{n_j\times k_j}, 
    \end{split}
  \end{equation*}
  which makes $\s$ even better than the sub-clustering $(\mx{C}^F_1, \ldots,
  \mx{C}^F_t)$ induced by the optimal $m$-dimensional clustering $\f$. Thus,
  \begin{align}
    \nonumber
    \frob{\mx{A} - \s \cdot \mx{A}}^2 & \leq \min_{\mx{M}}\ \frob{\mx{A} -
      (\mx{C}^F_1, \ldots,  \mx{C}^F_t,  \id^0) \cdot \mx{M}}^2\\ 
    \nonumber
    &\leq \frob{\mx{A} - ( \mx{C}^F_1, \ldots, \mx{C}^F_t, \id^0) (\mx{I}, \ldots,
    \mx{I}, \mx{C}^F_{t+1}, \ldots, \mx{C}^F_m) \cdot \mx{M}^F}^2\\ 
    \label{eq:10}
    &= \frob{\mx{A} - \f \cdot \mx{A}}^2, 
  \end{align}
  where $\mx{M}^F$ is the tensor of means for the optimal $m$-dimensional
  clustering. 
  Combining~(\ref{eq:14}) with~\eqref{eq:10} yields the final bound for
  the combined clustering $\C = \q_1^h$,
  \begin{equation*}
    J_m(\C) =  \frob{\mx{A} - \q_1^{h}\cdot \mx{A}}^2 \leq 2^h\alpha_t \frob{\mx{A} - \f \cdot
      \mx{A}}^2 = 2^h \alpha_t \jopt(m),
  \end{equation*}
  which completes the proof of the theorem.
\end{proof}

\subsection{Analysis: Theorem~\ref{thm:mway}, Metric case}
Now we present our proof of Thm.~\ref{thm:mway} for the case where
$\dist(x,y)$ is a metric, such as an $\ell_p$ distance or separable
Hilbertian metric.
For this case, recall that the tensor clustering
problem is
\begin{equation}
  \label{eq:17}
  \underset{(\mx{C}_1,\ldots,\mx{C}_m),\mx{M}}{\text{minimize}}
  J(\C) = \dist(\mx{A}, (\mx{C}_1, \ldots, \mx{C}_m) \cdot  \mx{M}),\quad\text{s.t. } \mx{C}_j \in \set{0,1}^{n_j\times k_j}. 
\end{equation}
Since in general the best representative $\mx{M}$ is not the mean tensor, we
cannot use the shorthand $\p \cdot \mx{A}$ for $\mx{M}$, so the proof is different
from the Euclidean case.

\begin{proof}
We will split the dimensions in a different way. Let $\bm{R}_i^\ell$ be
the collection of clusterings of dimensions $i, \ldots,
i+\ell-1$. $\rr_i^\ell$ combines the $\C_j$ in a manner analogous to how
$\q_i^l$ combines projection matrices. For simplicity, the proof here
is for clustering single dimensions at a time, but it generalizes in a
straightforward way to chunks of $t$ dimensions, leading to a factor
$2m/t$ instead of $2m$.

Let us first prove a relation for 
any subset of the last
$m-i+1$ dimensions, $\rr_i^1\rr_{i+1}^{m-i} = \rr_i^{m-i+1}$. 
Let $\mx{M}^\ell_i = \argmin_{\mx{X}} d(\mx{A}, \bm{R}^\ell_i\cdot \mx{X})$
be the optimal representatives for the clustering collections
$\rr_i^1$ and $\rr_{i+1}^{m-i}$, and
\begin{equation*}
  \widehat{\mx{M}}_i = \argmin_{\mx{X}} \dist(\bm{R}^1_i\mx{M}^1_i,
  \bm{R}^{1}_i\bm{R}^{m-i}_{i+1}\cdot \mx{X}), \quad   
  \mx{X} \in \mathbb{R}^{n_1 \times \ldots \times n_{i-1} \times k_i
    \times \ldots \times k_m}.
\end{equation*}

The index $\iota$ will run over dimension $i$, and the multi-indices,
$r$, $j$ over dimensions $1, \ldots, i-1$ and $i+1, \ldots, m$,
respectively. The indices $I$ and multi-indices $J$ refer to the
clusterings in  $\rr_i^1$ and $\rr_{i+1}^{m-i}$,
respectively. Since
$\widehat{\mx{M}}_i$ is the element-wise minimum, we have
\begin{align*}
  &\dist(\bm{R}^1_i \cdot \mx{M}^1_i, \bm{R}^1_i\bm{R}^{m-i}_{i+1}\cdot \widehat{\mx{M}}_i) =
  \sum_{I,J}\sum_{\iota \in I,r} \min_{\mu_{IJr}\in \mathbb{R}} \sum_{j \in J} \dist((\mu^1_\iota)_{Ijr}, \mu_{IJr})\\
  & \leq \sum_{I,J}\sum_{\iota \in I,r}\sum_{j \in J}
  \dist((\mu^1_\iota)_{Ijr}, (\mu^{m-i}_{i+1})_{\iota Jr})\quad=\quad \dist(\bm{R}^1_i\cdot \mx{M}^1_i, \bm{R}^{m-i}_{i+1}\cdot \mx{M^{m-1}_{i+1}}).
\end{align*}
We use this relation and the triangle inequality to break down
$\rr_1^m$ into its single-dimansional parts. We then relate the
objectives of these parts to the optimal single-dimensional
clusterings $\bm{S}^{1}_i$.
\begin{align}
  \nonumber
  &\min_{\mx{M^{m}}}\ \ \dist(\mx{A}, \bm{R}^1_1 \bm{R}^{m-1}_2 \cdot
  \mx{M}^{m})  \leq 
  \dist(\mx{A}, \bm{R}^1_1 \bm{R}^{m-1}_2\cdot \widehat{\mx{M}}_1)\\ 
  \nonumber
  &\qquad\leq \dist(\mx{A}, \bm{R}^1_1\cdot \mx{M}^1_1) + \dist(\bm{R}^1_1\cdot
  \mx{M}^1_1, \bm{R}^1_1\bm{R}^{m-1}_2\cdot \widehat{\mx{M}}_1) \\ 
  \nonumber
  &\qquad\leq  \dist(\mx{A}, \bm{R}^1_1 \cdot \mx{M}^1_1) + \dist(\bm{R}^1_1\cdot
  \mx{M}^1_1, \bm{R}^{m-1}_2\cdot \mx{M^{m-1}_2})\\
  \nonumber
  &\qquad\leq 2\dist(\mx{A}, \bm{R}^1_1\cdot \mx{M}^1_1) +
  \dist(\mx{A}, \bm{R}^{m-1}_2\cdot \mx{M}^{m-1}_2)\\
  \label{eq:repbreak}
  &\qquad\leq  2\dist(\mx{A}, \bm{R}^1_1\cdot \mx{M}^1_1) +
  2\dist(\mx{A}, \bm{R}^1_2\cdot \mx{M}^1_2) + \dist(\mx{A},
  \bm{R}^{m-1}_2\cdot \mx{M}^{m-1}_2)\\
  \nonumber
  &\qquad\leq \ldots\\
  \label{eq:sumsingle}
  &\qquad\leq 2\sum_{i=1}^m \dist(\mx{A}, \bm{R}^{1}_i\cdot
  \mx{M}^{1}_i) \qquad\leq 2\sum_{i=1}^m \alpha_1 \min_{\mx{X}}\dist(\mx{A},
  \bm{S}^{1}_i \cdot \mx{X}^1). 
\end{align}
For (\ref{eq:repbreak}), we applied the same steps as before to
$\rr^1_2$ and $\rr^{m-2}_3$, and then continued this breakdown, always
splitting off the first dimension. The last relation follows from the
1D approximation algorithm that was used. What is left is to bound
(\ref{eq:sumsingle}) by the objective for the optimal $m$-dimensional clustering
$\bm{F}\cdot \mx{M}_F = \bm{F}_1 \bm{F}_2 \ldots \bm{F}_m \cdot
  \mx{M}_F$. Note that, since non-clustered dimensions have identity
  matrices, the cluster parts commute: $\bm{F}_i \bm{F_j} \mx{X} =
  \bm{F}_j \bm{F_i} \mx{X}$.
Owing to the optimality of $\bm{S}^1_i$, we have
\begin{equation*}
  \min_{\mx{X}^1}\ \dist(\mx{A}, \bm{S}^1_i\cdot \mx{X}^1)\ \leq \ 
  \min_{\mx{Y}^1}\ \dist(\mx{A}, \bm{F}_i\cdot \mx{Y}^1)\ \leq \ 
  \min_{\mx{Y}^{m}}\ \dist(\mx{A}, \bm{F}^1_i (\bm{F}^1_1\ldots
  \bm{F}^1_{i-1} \bm{F}^1_{i+1}\ldots \bm{F}^1_m \cdot \mx{Y}^{m})) =
  \dist(\mx{A}, 
  \bm{F}\cdot \mx{M}_F)
\end{equation*}
for any term in the sum (\ref{eq:sumsingle}). Thus, it follows that
\begin{align*}
  \min_{\mx{M^{m}}}\ \ \dist(\mx{A}, \bm{R}^m_1 \cdot
  \mx{M}^{m}) \leq 2\sum_{i=1}^m \alpha_1 \min_{\mx{X}}\dist(\mx{A},
  \bm{S}^{1}_i \cdot \mx{X}^1) \leq 2m\alpha_1 \dist(\mx{A}, 
  \bm{F}\cdot \mx{M}_F),
\end{align*}
which completes the proof.
\end{proof}

\subsection{Theorem~\ref{thm:mway} with Bregman divergences}
Theorem~\ref{thm:mway} also applies to Bregman divergences, i.e., 
divergences that can be bounded in terms of squared Euclidean
distances and for which the best representative is the tensor of means
defined in Equation~(\ref{eq:28}) \cite{banerjee}. 

The \emph{Bregman divergence} $\breg(x, y)$ between scalars $x$ and $y$ is
defined as~\citep{breg67,censor:zenios}
\begin{equation}
  \label{eq:4}
  \breg(x, y) = f(x) - f(y) - f'(y)(x-y),
\end{equation}
for a given strictly convex function $f$. With $f = \half x^2$ the
divergence~\eqref{eq:4} reduces to the familiar Euclidean distance
$\half(x-y)^2$, while for $f(x) = x \log x$ it turns into the (generalized)
KL Divergence. For tensors, we extend Definition~\eqref{eq:4} by considering
\emph{separable} Bregman divergences, so that
\begin{equation*}
  \breg(\mx{X}, \mx{Y}) = \sum\nolimits_{i_1,\ldots,i_m} \breg(x_{i_1\ldots i_m}, y_{i_1\ldots i_m}).
\end{equation*}

Let $\sigma_U$ and $\sigma_L$ be
upper and lower bounds, respectively, with $\sigma_L > 0$, such that
\begin{equation}
  \label{eq:bounds}
  \sigma_L \breg(x,y) \leq \|x-y\|^2 \leq \sigma_U \breg(x,y)
\end{equation}
for all $x,y$ in the convex hull of the entries of the given tensor
$\mx{A}$. For KL-divergence, the data must then be bounded away from
zero. 

Since the means tensor is the best representative $\argmin_{\mx{X}}
\breg(\mx{A}, \bm{R}\cdot \mx{A})$ for a clustering $\bm{R}$, we again
use use projection matrices to express clusterings. Let
$\q_1^h$ be, as above, the full combination of projection matrices
from dimension-wise clustering, and $\bm{F} = \argmin_{\bm{Q}}
\breg(\mx{A},\bm{Q}\cdot \mx{A})$ the optimal $m$-dimensional tensor
clustering. Then we know that 
\begin{align}
  \nonumber
  \breg(\mx{A}, \q_1^h) &\leq \sigma_U \|\mx{A}, \q_1^h\|^2\\
  \label{eq:useL2}
  &\leq  \sigma_U 2^{\log_2 m/t} \max_j \|\mx{A} - \q_j^0 \cdot \mx{A}\|^2\\
  \nonumber
  &\leq \frac{\sigma_U}{\sigma_L}2^{\log_2 m/t}\max_j D(\mx{A}, \q_j^0 \cdot
  \mx{A})\\
  \label{eq:refFull}
  &\leq \frac{\sigma_U}{\sigma_L} 2^{\log_2 m/t}\breg(\mx{A}, \bm{F} \cdot
  \mx{A}),
\end{align}
so $\rho_d = \frac{\sigma_U}{\sigma_L}$.
Inequality~(\ref{eq:useL2}) follows from Lemma~\ref{lemm:rec}, and
Inequality~(\ref{eq:refFull}) from an argumentation analogous to
Equation~(\ref{eq:10}).

Curvature bounds as in (\ref{eq:bounds}) seem to be
necessary for Bregman divergences to
guarantee \emph{constant} approximation factors for the underlying 1D
clustering---this intuition is reinforced by the results of~\cite{mcgregor},
who avoided such curvature assumptions and had to be content with a
\emph{non-constant} $O(\log n)$ approximation factor for information theoretic
clustering.

\subsection{Implications}
\label{sec:impl}
To obtain concrete bounds for a variety of tensor clustering problems,
we can use Theorem~\ref{thm:mway} for $t=1$ or $t=2$ with existing 1D
approximation factors $\alpha_t$ from the
literature. Table~\ref{tab:guaran} summarizes the results.

\subsubsection{1D factors for Metric and Bregman clustering}
The $(1+\epsilon)$ approximation factor for $1D$ clustering by
\citet{ackerman} applies to all metrics. It leads to an
$m$-dimensional approximation factor of $\alpha_m=p(m/t)(1+\epsilon)$.  
\citet{dasv07} prove a guarantee in expectation of $\alpha_1 = 8(\log
K + 2)$ for $K$ clusters with Euclidean k-means, resulting in an
expected $\alpha_m = 8p(m/t)(\log K + 2)$.

For Bregman clustering, we arrive at similar results with the
approximation factor
by \citet{acker2} or the extension of \citep{dasv07} in
\citep{ecmlbreg,kmeanstr}.

\subsubsection{Hilbertian metrics}
A special example of metrics are Hilbertian
metrics~\cite{scsm01,hebo05} that arise from conditionally positive
definite (CPD) kernels. 
A real valued function $C: \mathcal{S}
× \mathcal{S} \mapsto
\mathbb{R}$ is called a {\em conditionally positive definite} (CPD)
kernel
on $\mathcal{S}$  if for
any positive integer $n$, any choice of $n$ elements $x_i \in
\mathcal{S}, [i]_1^n$ ($[i]_1^n \equiv i=1,\ldots,n$) and any choice
of $n$ reals $u_i \in \mathbb{R}$ such that $\sum_i u_i = 0$, we have $
\sum_{i,j =0}^n u_i u_j C(x_i, x_j) \geq 0$~\cite{scsm01,becr84}.
The
following remarkable result~\cite{scho38} connects CPD kernels and
Hilbertian metrics, i.e., metrics which can be isometrically embedded
in Hilbert space: There exists a Hilbert space ${\cal H}$ of real-valued
functions on $\mathcal{S}$, and a mapping $\Phi : \mathcal{S} \mapsto {\cal H}$ such
that
\begin{equation*}
\| \Phi(\x) - \Phi(\y) \|^2 = -C(\x,\y) + \frac{1}{2}(C(\x,\x) + C(\y,\y)) = d_C(\x,\y)~,
\end{equation*}
if and only if $C(\cdot,\cdot)$ is a CPD kernel. Hence, given a CPD kernel $C$,
one can construct a Hilbertian metric $d_C(\x,\y)$ which behaves like the squared
Euclidean distance in the Hilbert space. The corresponding kernel is
$K(x,y) = \frac{1}{2}(C(x,y) - C(x,a) - C(y,a) + C(a,a))$ for some
fixed $a \in \mathcal{S}$. 

Here, we choose $\mathcal{S} \subseteq \reals$ and define the distance
of tensors $\mx{X}, \mx{Y}$ as
\begin{equation*}
  d_C(\mx{X},\mx{Y}) = \sum_{i_1,\ldots,i_m}
  d_C(x_{i_1,\ldots,i_m},x_{i_1,\ldots,i_m}). 
\end{equation*}

Since the argument by \cite{dasv07} for their kmeans++ is independent
of the dimensionality, it can be generalized from Euclidean distance 
to distances in a Hilbert space.
\begin{lemma}[1D Hilbertian Metric Clustering]
\label{lem:hilbert1}
For any 1D clustering with a Hilbertian metric $d_C$, one can construct a kmeans++ based
initialization followed by iterative updates using kernel k-means such that if 
$\C$ is the final clustering, then
\begin{equation}
\E[J(\C)] \leq 8 (log K + 2) J_{OPT}~.
\end{equation}
\end{lemma}
\begin{proof}
Using $d_C(x,y) = \|\Phi(x) - \Phi(y)\|^2$, we can use the
initialization by \cite{dasv07} in the Hilbert space on the mapped
data points $\Phi(x)$, since it only depends on squared Euclidean
distances  or inner products, independent of the dimensionality of the
space. 
Finally, the objective function can always be
improved by running kernel kmeans starting from the kmeans++
initialization. 
\end{proof}

Together with Theorem~\ref{thm:mway}, Lemma~\ref{lem:hilbert1} directly
leads to a tensor clustering guarantee for Hilbertian metrics:
\begin{equation} \label{eq:hilbert3}
  \E[J(\C)] \leq 8m (\log K^* +2)\jopt(m), 
\end{equation} 
where  $K^* = \max_{1\leq j\leq m} k_j$ is the maximum number of clusters
  across all dimensions.

\subsubsection{2D factor for binary $\ell_1$ clustering}
Applying the results of \cite{coclus1} for binary matrices as
$\alpha_2$ yields the slightly stronger bound for $\ell_1$ tensor
clustering: 
\begin{equation*}
  J(\C) \leq 3^{\log_2(m) - 1}(1+\sqrt{2})\alpha_1\jopt(m).
\end{equation*}

\begin{table}[h]
  \centering
  \caption{Approximation guarantees for Tensor Clustering Algorithms. $K^*$
    denotes the maximum number of clusters, i.e., $K^* = \argmax_j k_j$; $c$ is
    some constant.}
  \label{tab:guaran}
  \begin{tabular}{l|l|l}
    Problem Name  & Approx. Bound & Proof\\
    \hline
    Metric tensor clustering & $J(\C) \leq m(1+\epsilon)\jopt(m)$ &
    Thm.~\ref{thm:mway} + \cite{ackerman}\\
    Bregman tensor clustering & $\E[J(\C)] \leq 8mc(\log K^* +2)\jopt(m)$
    &\eqref{eq:bounds}, Thm.~\ref{thm:mway} + \cite{ecmlbreg,kmeanstr} (using \cite{dasv07})\\
    Bregman tensor clustering  & $J(\C) \leq m \sigma_U\sigma_L^{-1}(1+\epsilon)\jopt(m)$
    &\eqref{eq:bounds}, Thm.~\ref{thm:mway} + \cite{acker2}\\
    Bregman co-clustering & Above two results with $m=2$ & as above\\
    Hilbertian metrics & $\E[J(\C)] \leq 8m (\log K^* +2)\jopt(m)$ &
    Thm.~\ref{thm:mway} + Lemma~\ref{lem:hilbert1}\\
    \hline
  \end{tabular}
\end{table}

\section{Experiments}
\label{sec:experiments}
Our bounds depend strongly on the approximation factor $\alpha_t$ of an underlying
$t$-dimensional clustering method. In our experiments, we study this close
dependence for $t=1$, wherein we compare the tensor clusterings arising from
different 1D methods of varying sophistication. Keep in mind that the
comparison of the 1D methods is to see their impact on the tensor clustering
built on top of them.

Our experiments reveal that the empirical approximation
factors are usually smaller than the theoretical bounds, and these factors
depend on statistical properties of the data. We also observe the linear
dependence of the \cotec\ objectives on the associated 1D objectives, as
suggested by Thm.~\ref{thm:mway} (for Euclidean) and Table~\ref{tab:guaran}
(2nd row, for KL-Divergence). 

Further comparisons show that in  practice, \cotec\ is competitive
with 
a greedy heuristic SiTeC (\textbf{Si}multaneous \textbf{Te}nsor
\textbf{C}lustering), which \emph{simultaneously} takes \emph{all}
dimensions into account, but lacks theoretical guarantees.
As expected, initializing SiTeC with \cotec\ yields lower final objective
values using fewer ``simultaneous'' iterations.

Regarding divergences, we focus on Euclidean distance and KL-divergence to
test \cotec.  To study the effect of the 1D method, we use two seeding methods
for each divergence, uniform and distance-based drawing. The latter seeding
ensures 1D approximation factors for $\E[J(\C)]$ by \cite{dasv07} for
Euclidean clustering and by \cite{ecmlbreg,kmeanstr} for KL-divergence.

We use each seeding by itself and as an initialization for k-means to get four
1D methods for each divergence. We refer to the \cotec\ combination of the
corresponding independent 1D clusterings by abbreviations:
\begin{description}
\setlength{\itemsep}{-2pt}
\item[r:] Randomly (uniformly) sample centers from the data points;
  assign each point to its closest center.
\item[s:] Sample centers using distance-specific
  seeding~\citep{dasv07,ecmlbreg,kmeanstr}; assign each point to its closest
  center.
\item[rk:] Initialize Euclidean or Bregman k-means with `\textbf{r}'.
\item[sk:] Initialize Euclidean or Bregman k-means with `\textbf{s}'.
\end{description}

The SiTeC method we compare to is the minimum sum-squared
residue co-clustering of~\cite{suv.coclus} for Euclidean distances in
2D, and a generalization of Algorithm~1 of~\cite{coclust}
for 3D and Bregman 2D clustering.
Additionally, we initialize  SiTeC
with the outcome of each of the four \cotec\ variants, which yields four
versions (of SiTeC), namely,

We compare the four versions of CoTeC to SiTeC, an algorithm without guarantees
that considers the groupings in all dimensions together. For Euclidean
distances in 2D, we use the minimum sum-squared residue co-clustering
of~\cite{suv.coclus} as SiTeC, while for Euclidean 3D and Bregman tensor clustering,
we generalize Algorithm~1 of~\cite{coclust}.  Initializing SiTeC with each
one of the above schemes results in another four variants:
\begin{description}
  \setlength{\itemsep}{-2pt}
\item[rc:] SiTeC initialized with the results of `\textbf{r}'
\item[sc:] SiTeC initialized with the results of `\textbf{s}'
\item[rkc:]  SiTeC initialized with the results of `\textbf{rk}'
\item[skc:]  SiTeC initialized with the results of `\textbf{sk}'
\end{description}
These variants inherit the
guarantees of \cotec, as they monotonically decrease the objective value.

\subsection{Experiments on synthetic data}
\label{sec:artificial-data}
For a controlled setting with synthetic data, we generate tensors $\mx{A}$ of
size $75 \times 75 \times 50$ and $75 \times 75$, for which we randomly choose a $5 \times 5 \times 5$
tensor of means $\mx{M}$ and cluster indicator matrices $\mx{C}_i \in \{0,1\}^{n_i
  \times 5}$.  For clustering with Euclidean distances we add Gaussian noise (from
${\cal N}(0,\sigma^2)$ with varying $\sigma$) to $\mx{A}$, while for KL-Divergences
we use the sampling method of~\cite{coclust} with varying noise.

For each noise-level to test, we repeat the 1D seeding 20 times on each of
five generated tensors and average the resulting 100 objective values.  To
estimate the approximation factor $\alpha_m$ on a tensor, we divide the achieved
objective $J(\C)$ by the objective value of the ``true'' underlying tensor
clustering.
Figure~\ref{fig:3dnoiP} shows the empirical approximation factor $\hat{\alpha}_m$ for
Euclidean distance and KL-Divergence. Qualitatively, the plots for tensors of
order 2 and 3 do not differ.

\begin{figure}
  \centering
  \includegraphics[height=2.1in,width=0.45\textwidth]{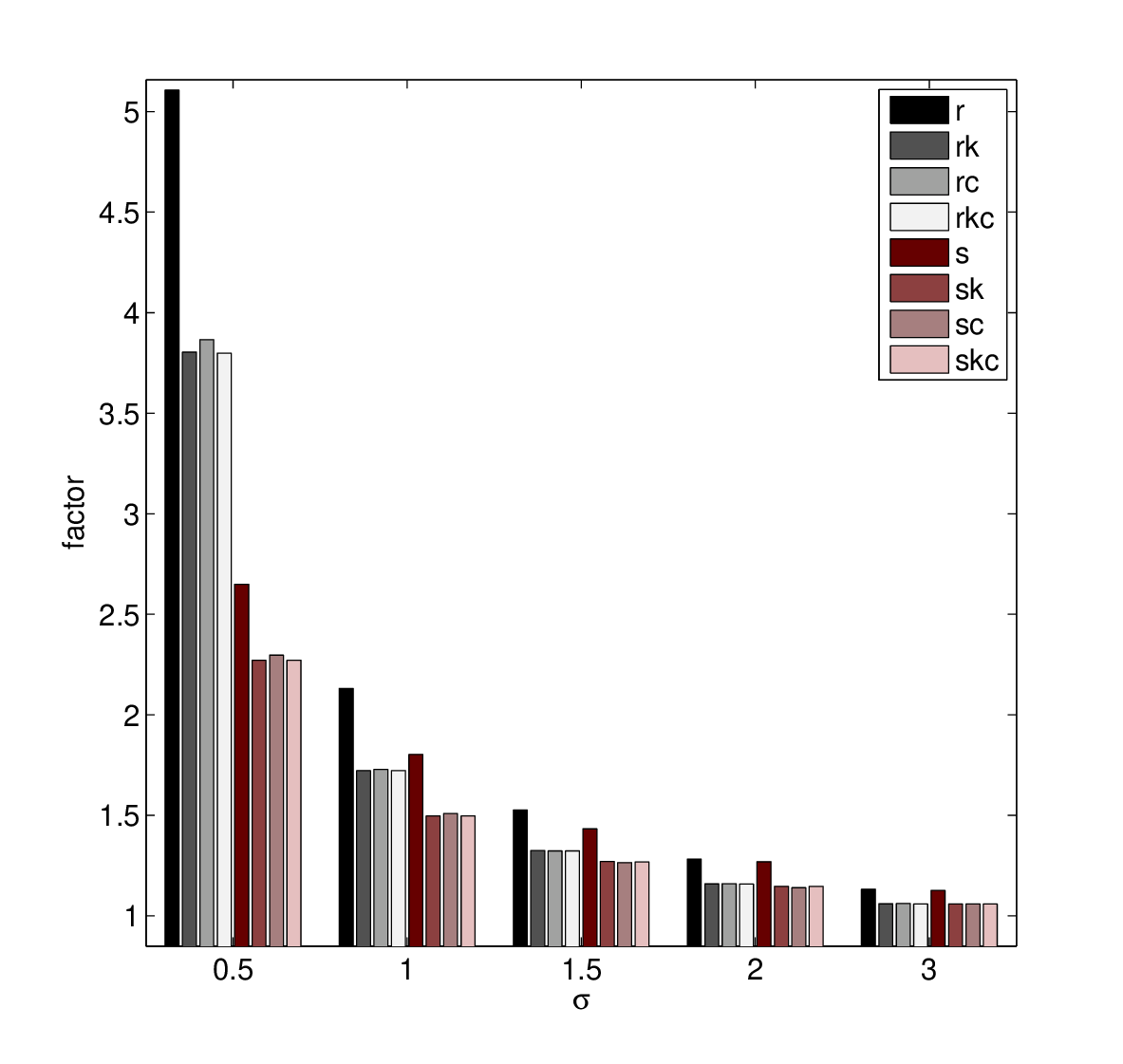}
  \includegraphics[height=2.1in,width = 0.45\textwidth]{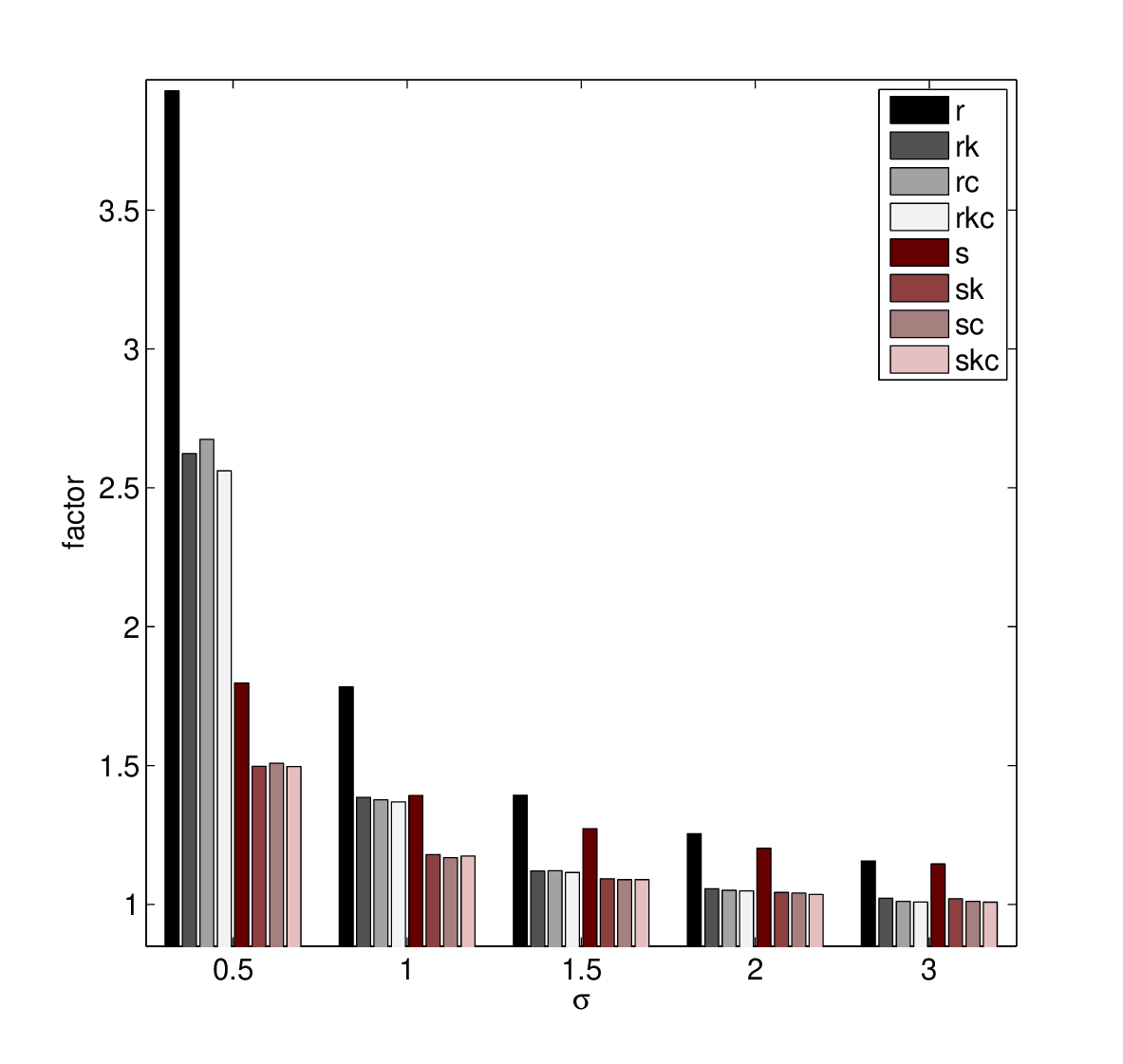}\\
  \includegraphics[height=2.1in,width = 0.45\textwidth]{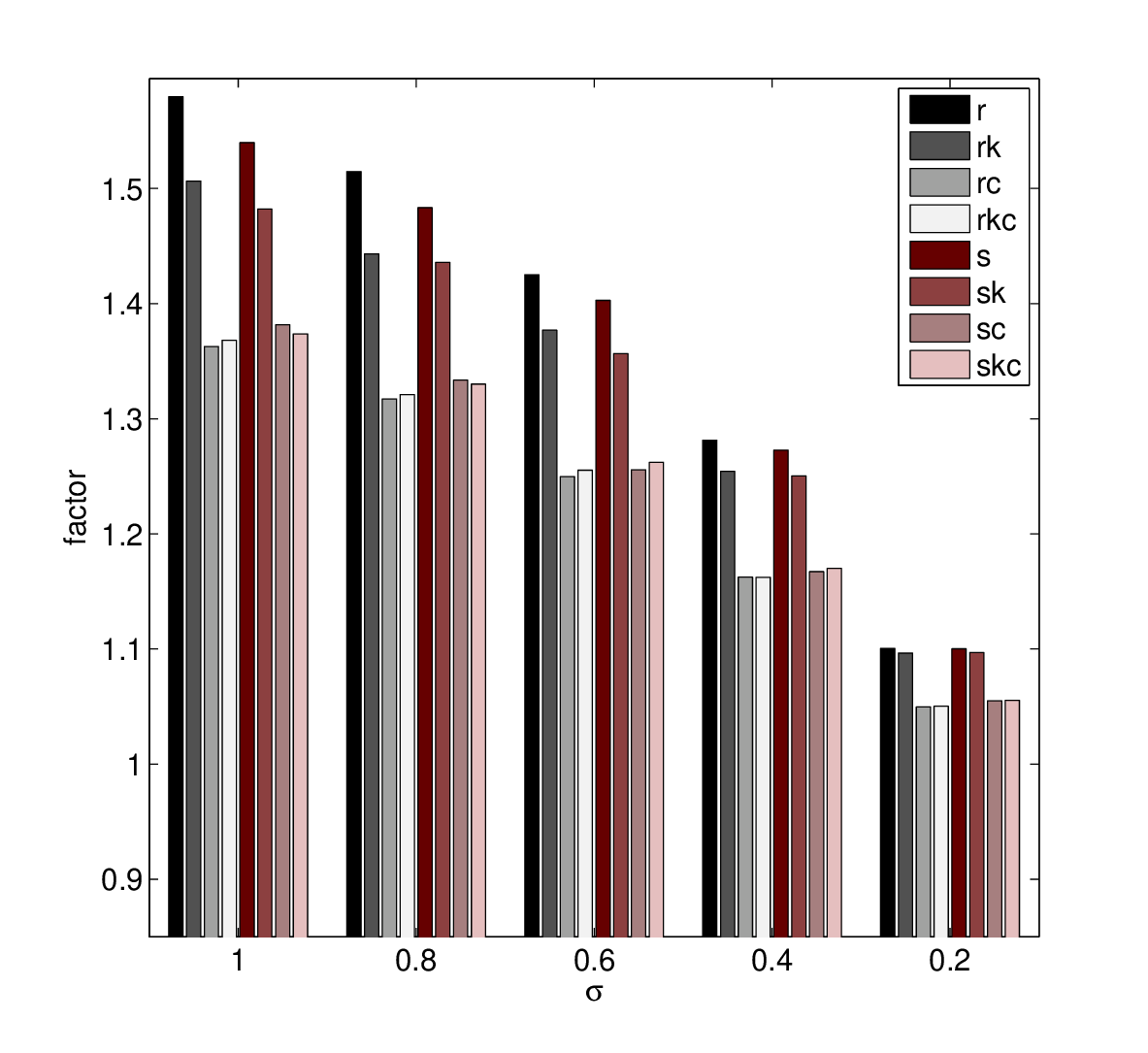}
  \includegraphics[height=2.1in,width = 0.45\textwidth]{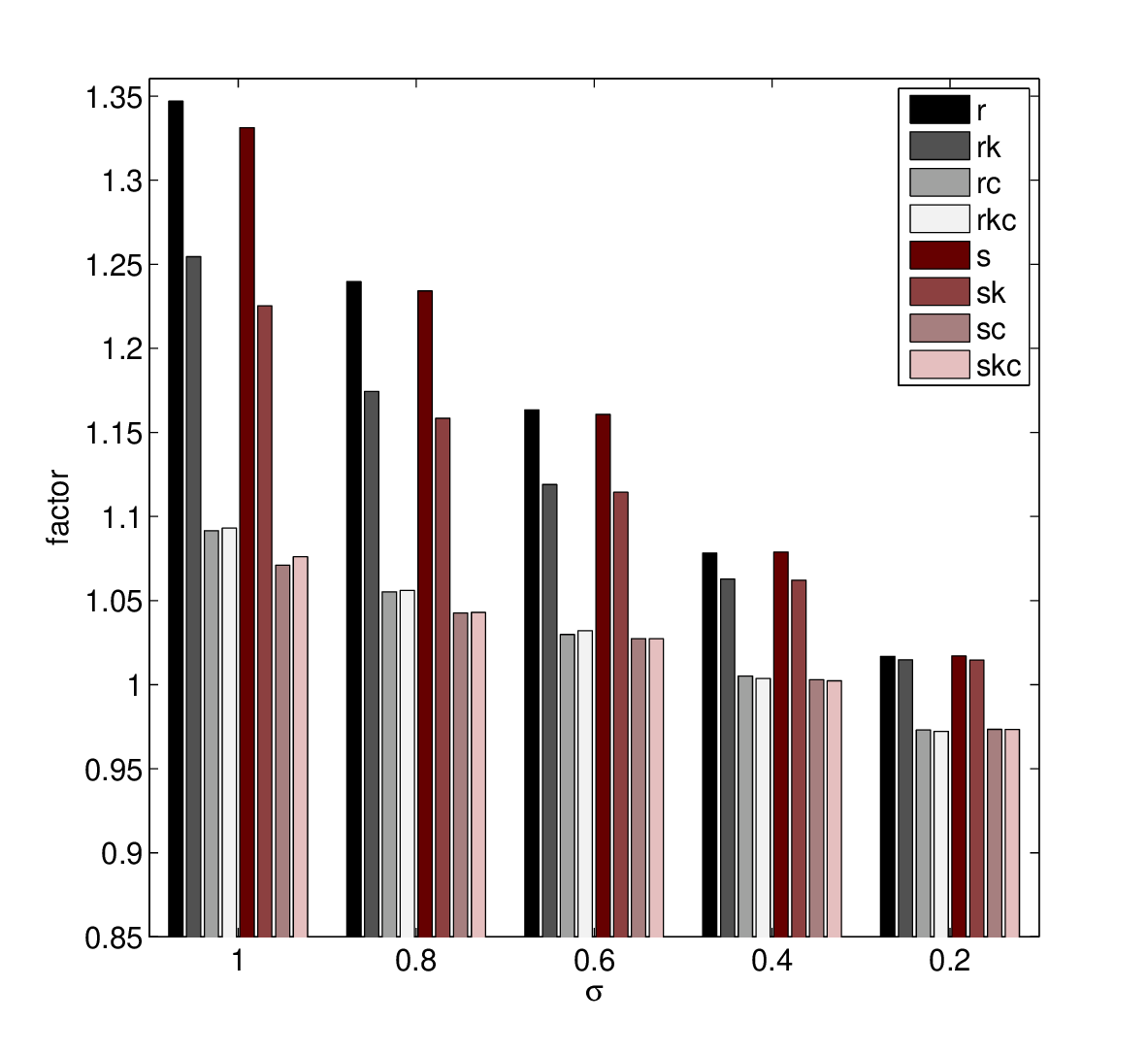}
\vspace{-8pt}
  \caption{\small Approximation factors for 3D clustering (left) and co-clustering
    (right) with increasing noise. Top row: Euclidean distances, bottom
    row: KL Divergence. The $x$ axis shows $\sigma$, the $y$ axis the
    empirical approximation factor.}
  \label{fig:3dnoiP}
\vspace{-2pt}
\end{figure}

In all settings, the empirical factor remains below the theoretical factor.
The reason for decreasing approximation factors with higher noise could be lower
accuracy of the estimates of $J(C)$ on the one hand, and more similar objective values
for all clusterings on the other hand.
With low
noise, distance-specific seeding \textbf{s} yields better results than uniform
seeding \textbf{r}, and adding k-means on top (\textbf{rk},\textbf{sk})
improves the results of both.  With Euclidean distances, CoTeC with
well-initialized 1D $k$-means (\textbf{sk}) competes with SiTeC. For
KL-divergence, though, SiTeC still improves on \textbf{sk}, and with high
noise levels,
1D $k$-means does not help: both \textbf{rk} and \textbf{sk} are as good as their
seeding only counterparts.

In summary, the empirical approximation factor does depend on the data, but
in general seems to be lower than the theoretical worst-case value.

\subsection{Experiments on real data}
We further assess the behavior of \cotec\ on a number of
real-world gene expression data sets\footnote{We thank Hyuk Cho for kindly
  providing us the preprocessed data.}.

The first three of our data sets, \emph{Bcell} ($1332 × 62$),
\emph{AllAml} ($2088 × 72$) and \emph{Breast} ($21906 × 77$) are
gene expression microarray data sets, and described in detail in
\cite{kluger03}.
\emph{Bcell} is a lymphoma microarray dataset of chronic lymphotic
leukemia, diffuse large Bcell leukemia and follicular lymphoma. During
preprocessing only those genes were selected whose minimum
expression level was above $e^{-1000}$. Microarray data for B-cell and
T-cell acute lymphocytic leukemia and acute myelogenous leukemia is
collected in \textit{AllAml}. Our data matrix is restricted to those genes
whose ratio of maximum to minimum  expression exceeds 10 and for whom
the difference between maximum and minimum expression was at least
1000. \textit{Breast} refers to breast cancer data. The gene selection was
the same as for \textit{Bcell}. 

The remaining two data sets are cancer microarray matrices
from~\citep{cho08}.%
\textit{Leukemia} ($3571 × 72$)  \cite{hyukref9} is data
from acute lymphoblastic leukemia or acute myeloid leukemia, 
and \textit{Mll} ($2474 × 72$) \cite{hyukref48}
includes data 
from three types of leukemia (ALL, AML, MLL).  

Even though the data sets have labeled column clusters, we do not compare
clustering results with the true labels, as the algorithm and its guarantees
hold merely for the clustering objective function, which may not exactly
agree with the true labels. Moreover, we aim for a co-clustering result
and not single-dimensional clusterings, and the labels are available 
for only one of the dimensions.

For each data set, we repeat the sampling of centers 30 times and average
the resulting objective values. Tables~\ref{tab:bcell} to
\ref{tab:mll} show detailed reults.
Panel (i) displays the objective value for the
simplest CoTeC, \textbf{r}, as a baseline, and the relative improvement
achieved by the other methods. The methods are encoded as \textbf{x},
\textbf{xk}, \textbf{xc}, \textbf{xkc}, where \textbf{x} stands for
\textbf{r} or \textbf{s}, depending on the row in the said table.

Overall, the improvements obtained via the approximation algorithm do depend
on the dataset under consideration and the number of clusters sought. In
general, the improvements are lower for the bispherically normalized data
(e.g., that of~\citep{cho08}) than for the other data sets.

For both distances, using 1D k-means on top of the seeding generally
improves on the combined co-clustering.
The combination method seems particularly competitive for Euclidean
distances. On the \emph{Bcell} data (Table~\ref{tab:bcell}), the
\textbf{s} variant of CoTeC (without k-means) can be \emph{as good as}
SiTeC  \textbf{r} initialization.  The distance-specific seeding (\textbf{s})
gains compared to uniform seeding as the 
clusters become smaller.  For \textit{Bcell} and \textit{Breast}
(Table~\ref{tab:bcell}), the combination of 1D k-means 
clusterings (\textbf{rk} and \textbf{sk}) slightly outperforms the
SiTeC variants \textbf{rc} and \textbf{sc}). 

Turning to KL Divergences, the impact of the 1D method varies with the
 data, as for Euclidean distance. Both 1D k-means and better seeding
mostly improve the overall outcome.
We observe the highest improvements on the \textit{AllAml}
data set. With KL Divergences, SiTeC is almost always at least a bit
better than CoTeC.

Besides improving the final result, a good initialization aids SiTeC in yet
another way: the average number of iterations it takes to converge
decreases, at times to even less than half the reference value.

Overall, the experiments demonstrate that the combination of good
single-dimensional clusterings can already lead to reasonable
co-clusterings in practice,
which can at times be as good as the result of a simultaneous biclustering
method. Used as an initialization, the CoTeC results improve the
outcome of SiTeC and reduce the number of ``simultaneous'' iterations.

\begin{table}[htbp]
  \caption{\small (i) Improvement of CoTeC and SiTeC variants upon
    `\textbf{r}' in \%; the respective reference value ($J_2$ for
    `\textbf{r}') is shaded in gray. 
    (ii) Average
    number of SiTeC iterations. } 
  \label{tab:bcell}
\vspace{1.5mm}
  \begin{tabular}{@{\hspace{-2pt}}l@{\hspace{1pt}}l}
    {\small
      \begin{tabular}{lc@{\hspace{2pt}}c@{\hspace{2pt}}c|rrrr}
        \multicolumn{8}{c}{\bf Bcell, Euc.}\\
      (i) &  $k_1$ & $k_2$ & & \cc{x} & \cc{xk} & \cc{xc} & \cc{xkc}\\ \cline{2-8} 
      &  5 & 3 & r &
        \gc{$6.00\cdot 10^5$} & $20.98$ & $18.37$ & $26.44$\\ 
      &  &  & s  & $8.52$ & $24.97$ & $22.83$ & $29.53$\\ \cline{2-8} 
      &  5 & 6 & r &  \gc{$5.94\cdot 10^5$} & $30.68$ & $26.09$ & $34.72$\\
      &  &  & s  & $16.97$ & $33.35$ & $32.06$ & $37.33$\\ \cline{2-8}  
      &  20 & 3 & r & \gc{$5.75\cdot 10^5$} & $31.66$ & $20.05$ & $33.05$\\ 
      &  &  & s  & $18.83$ & $32.24$ & $24.61$ & $33.36$\\ \cline{2-8} 
      &  20 & 6 & r & \gc{$5.56\cdot 10^5$} & $49.13$ & $35.26$ & $50.37$\\ 
      &  &  & s  & $34.97$ & $50.55$ & $43.93$ & $51.66$\\ \cline{2-8} 
      &  50 & 3 & r & \gc{$5.63\cdot 10^5$} & $31.10$ & $14.77$ & $31.76$\\ 
      &  &  & s  & $15.25$ & $32.58$ & $19.14$ & $33.17$\\ \cline{2-8} 
      &  50 & 6 & r & \gc{$5.18\cdot 10^5$} & $47.55$ & $34.63$ & $48.41$\\ 
      &  &  & s  & $36.22$ & $49.83$ & $43.77$ & $50.55$\\%
      \end{tabular}
    } &
    {\small
      \begin{tabular}{lc@{\hspace{2pt}}c@{\hspace{2pt}}c|rrrr}
        \multicolumn{8}{c}{\bf Bcell, KL}\\
     (i) &   $k_1$ & $k_2$ & & \cc{x} & \cc{xk} & \cc{xc} & \cc{xkc}\\ \cline{2-8} 
      &  5 & 3 & r & \hc{3.73}{-1} & $15.01$ & $20.87$ & $21.13$\\ 
      &  &  & s  & $1.53$ & $14.31$ & $20.43$ & $20.26$\\ \cline{2-8}  
      &  5 & 6 & r & \hc{3.60}{-1} & $15.76$ & $21.23$ & $21.62$\\ 
      &  &  & s  & $3.24$ & $16.22$ & $21.37$ & $21.21$\\ \cline{2-8} 
      &  20 & 3 & r & \hc{3.37}{-1} & $17.59$ & $22.23$ & $23.26$\\ 
      &  &  & s  & $10.54$ & $18.44$ & $22.99$ & $22.98$\\ \cline{2-8} 
      &  20 & 6 & r & \hc{3.15}{-1} & $18.62$ & $24.51$ & $25.43$\\ 
      &  &  & s  & $11.76$ & $20.52$ & $25.69$ & $26.23$\\ \cline{2-8} 
      &  50 & 3 & r & \hc{3.20}{-1} & $15.70$ & $20.12$ & $21.07$\\ 
      &  &  & s  & $9.61$ & $17.24$ & $20.85$ & $21.33$\\ \cline{2-8} 
      &  50 & 6 & r & \hc{2.85}{-1} & $16.38$ & $21.61$ & $22.57$\\ 
      &  &  & s  & $11.86$ & $18.63$ & $23.24$ & $23.13$\\
      \end{tabular}
    } \\
    \phantom{blank} & \phantom{line}\\
    {\small
      \hspace{1pt}
      \begin{tabular}[htbp]{@{\hspace{2pt}}l@{\hspace{2pt}}c@{\hspace{2pt}}c@{\hspace{2pt}}|@{\hspace{4pt}}r@{\hspace{8pt}}r@{\hspace{8pt}}r@{\hspace{8pt}}r}
    (ii) &    $k_1$ & $k_2$ & \multicolumn{1}{c}{rc} & \cc{rkc} & \cc{sc} & \cc{skc}\\ \cline{2-7}
     &   5 & 3& $ 11.9±3.3$& $ 3.3±0.7$& $ 6.1±2.8$& $ 3.5±0.7$\\
     &   5 & 6& $ 11.9±2.6$& $ 3.7±1.7$& $ 6.6±2.4$& $ 3.3±1.3$\\
     &   20 & 3& $ 7.0±1.4$& $ 2.0±0.2$& $ 3.9±1.0$& $ 2.2±0.5$\\
     &   20 & 6& $ 11.3±2.3$& $ 2.6±0.8$& $ 5.1±2.0$& $ 2.7±0.7$\\
     &   50 & 3& $ 6.2±1.9$& $ 2.0±0.0$& $ 3.5±2.0$& $ 2.0±0.0$\\
     &   50 & 6& $ 8.1±2.1$& $ 2.1±0.3$& $ 4.1±1.6$& $ 2.0±0.0$\\
      \end{tabular}
    } &
    {\small
      \begin{tabular}[htbp]{@{\hspace{2pt}}l@{\hspace{2pt}}c@{\hspace{2pt}}c@{\hspace{2pt}}|@{\hspace{4pt}}r@{\hspace{8pt}}r@{\hspace{8pt}}r@{\hspace{8pt}}r}
     (ii) &   $k_1$ & $k_2$ & \cc{rc} & \cc{rkc} & \cc{sc} & \cc{skc}\\ \cline{2-7}
      &  5 & 3& $ 10.1±3.0$& $ 7.2±3.0$& $ 11.1±4.3$& $ 7.2±3.5$\\
      &  5 & 6& $ 10.8±3.1$& $ 8.1±3.4$& $ 8.7±2.9$& $ 6.8±3.3$\\
      &  20 & 3& $ 10.6±2.8$& $ 7.5±2.0$& $ 7.4±1.8$& $ 7.0±2.2$\\
      &  20 & 6& $ 12.6±3.4$& $ 8.8±2.9$& $ 8.4±2.1$& $ 8.1±2.0$\\
      &  50 & 3& $ 9.1±2.3$& $ 6.2±1.3$& $ 6.9±1.8$& $ 6.0±1.3$\\
      &  50 & 6& $ 10.5±1.8$& $ 7.7±2.1$& $ 8.1±2.3$& $ 6.9±1.0$\\
      \end{tabular}
    }
    \end{tabular}
    \vskip 10pt
  \begin{tabular}{@{\hspace{-2pt}}l@{\hspace{1pt}}l}
    {\small
      \begin{tabular}{lc@{\hspace{2pt}}c@{\hspace{2pt}}c|rrrr}
        \multicolumn{8}{c}{\bf Breast, Euc}\\ 
 (i) &       $k_1$ & $k_2$ & & \cc{x} & \cc{xk} & \cc{xc} & \cc{xkc}\\ \cline{2-8}
     &   5 & 2 & r & \hc{1.43}{5} & $22.96$ & $20.48$ & $24.47$\\ 
     &   &  & s  & $2.69$ & $21.92$ & $19.42$ & $24.32$\\ \cline{2-8} 
     &   5 & 4 & r & \hc{1.42}{5} & $26.49$ & $25.85$ & $27.30$\\ 
     &   &  & s  & $10.38$ & $26.72$ & $26.67$ & $27.95$\\ \cline{2-8}
     &   10 & 2 & r & \hc{1.41}{5} & $22.13$ & $15.46$ & $25.26$\\ 
     &   &  & s  & $7.77$ & $21.66$ & $19.20$ & $25.09$\\ \cline{2-8}
     &   10 & 4 & r & \hc{1.37}{5} & $26.36$ & $24.09$ & $28.93$\\ 
     &   &  & s  & $9.79$ & $26.87$ & $26.44$ & $29.90$\\ \cline{2-8}
     &   20 & 2 & r & \hc{1.41}{5} & $22.46$ & $10.42$ & $26.21$\\ 
     &   &  & s  & $8.16$ & $22.54$ & $19.43$ & $26.16$\\ \cline{2-8}
     &   20 & 4 & r & \hc{1.37}{5} & $27.95$ & $23.44$ & $31.71$\\ 
     &   &  & s  & $10.55$ & $28.31$ & $25.83$ & $32.44$\\ 
      \end{tabular}
    } &
    {\small
      \begin{tabular}{lc@{\hspace{2pt}}c@{\hspace{2pt}}c|rrrr}
        \multicolumn{8}{c}{\bf Breast, KL}\\
   (i) &     $k_1$ & $k_2$ & & \cc{x} & \cc{xk} & \cc{xc} & \cc{xkc}\\ \cline{2-8}
     &   5 & 2 & r & \hc{2.70}{-2} & $8.08$ & $12.81$ & $12.23$\\ 
     &   &  & s  & $1.77$ & $7.98$ & $13.19$ & $12.38$\\ \cline{2-8} 
     &   5 & 4 & r & \hc{2.67}{-2} & $11.88$ & $17.56$ & $17.31$\\ 
     &   &  & s  & $3.60$ & $11.95$ & $18.10$ & $18.29$\\ \cline{2-8} 
     &   10 & 2 & r & \hc{2.66}{-2} & $8.01$ & $11.44$ & $12.37$\\ 
     &   &  & s  & $2.45$ & $7.96$ & $12.34$ & $12.46$\\ \cline{2-8} 
     &   10 & 4 & r & \hc{2.59}{-2} & $11.17$ & $16.54$ & $17.92$\\ 
     &   &  & s  & $4.97$ & $13.53$ & $19.50$ & $19.31$\\ \cline{2-8} 
     &   20 & 2 & r & \hc{2.63}{-2} & $6.27$ & $9.72$ & $9.95$\\ 
     &   &  & s  & $2.93$ & $8.78$ & $11.69$ & $11.61$\\ \cline{2-8} 
     &   20 & 4 & r & \hc{2.56}{-2} & $11.73$ & $17.42$ & $17.78$\\ 
     &   &  & s  & $3.45$ & $12.21$ & $17.51$ & $17.45$\\ 
      \end{tabular}
    }\\
    \phantom{blank} & \phantom{line}\\
    {\small
      \begin{tabular}[htbp]{@{\hspace{2pt}}l@{\hspace{2pt}}c@{\hspace{2pt}}c@{\hspace{2pt}}|@{\hspace{4pt}}r@{\hspace{8pt}}r@{\hspace{8pt}}r@{\hspace{8pt}}r}
 (ii) &   $k_1$ & $k_2$ & \cc{rc} & \cc{rkc} & \cc{sc} & \cc{skc}\\         \cline{2-7}
   &   5 & 2& $ 4.6±2.4$& $ 1.2±0.4$& $ 4.0±1.6$& $ 1.8±0.4$\\
     &   5 & 4& $ 4.9±1.8$& $ 1.0±0.2$& $ 3.0±0.9$& $ 1.2±0.5$\\
     &   10 & 2& $ 3.4±1.4$& $ 2.0±0.2$& $ 2.6±1.0$& $ 2.0±0.0$\\
     &   10 & 4& $ 4.3±1.8$& $ 2.0±0.5$& $ 3.0±0.9$& $ 2.1±0.3$\\
     &   20 & 2& $ 2.9±1.3$& $ 2.0±0.0$& $ 2.7±1.0$& $ 2.0±0.0$\\
     &   20 & 4& $ 3.9±1.3$& $ 2.1±0.3$& $ 3.4±1.8$& $ 2.0±0.2$\\
      \end{tabular}      
    } & 
    {\small
      \begin{tabular}{@{\hspace{2pt}}l@{\hspace{2pt}}c@{\hspace{2pt}}c@{\hspace{2pt}}|@{\hspace{4pt}}r@{\hspace{8pt}}r@{\hspace{8pt}}r@{\hspace{8pt}}r}
   (ii) &     $k_1$ & $k_2$ & \cc{rc} & \cc{rkc} & \cc{sc} & \cc{skc}\\ \cline{2-7}
     &   5 & 2& $ 5.2±2.0$& $ 3.6±2.0$& $ 4.9±2.6$& $ 3.1±1.8$\\
     &   5 & 4& $ 5.6±1.8$& $ 3.6±1.9$& $ 4.4±1.2$& $ 3.5±1.4$\\
     &   10 & 2& $ 4.0±1.8$& $ 2.5±1.0$& $ 4.4±2.8$& $ 2.7±1.7$\\
     &   10 & 4& $ 5.1±1.4$& $ 4.0±1.7$& $ 5.2±1.7$& $ 3.7±1.3$\\
     &   20 & 2& $ 3.6±1.8$& $ 2.3±0.9$& $ 3.2±1.5$& $ 2.1±0.5$\\
     &   20 & 4& $ 5.2±1.9$& $ 3.5±1.8$& $ 4.3±1.6$& $ 2.8±1.2$\\
      \end{tabular}
    }
  \end{tabular}
\vspace{-6pt}
\end{table}

\begin{table}[htbp]
\caption{\small (i) Improvement of CoTeC and SiTeC variants upon
    `\textbf{r}' in \%; the respective reference value ($J_2$ for
    `\textbf{r}') is shaded in gray. 
    (ii) Average
    number of SiTeC iterations. } 
\vspace{1.5mm}

  \begin{tabular}{@{\hspace{-2pt}}l@{\hspace{1pt}}l}
    {\small
    \begin{tabular}{lc@{\hspace{2pt}}c@{\hspace{2pt}}c|rrrr}
      \multicolumn{8}{c}{\bf AllAml, Euc.}\\ 
   (i) &     $k_1$ & $k_2$ & & \cc{x} & \cc{xk} & \cc{xc} & \cc{xkc}\\
   \cline{2-8}
  &    5 & 3 & r & \hc{6.06}{11} & $49.26 $ & $49.19 $ & $50.54 $\\ 
  &    &  & s  & $40.63 $ & $48.62 $ & $50.27 $ & $50.71 $\\ \cline{2-8} 
  &    10 & 3 & r & \hc{5.31}{11} & $47.02 $ & $47.01 $ & $48.69 $\\ 
  &    &  & s  & $40.83 $ & $48.65 $ & $49.51 $ & $50.10 $\\ \cline{2-8}
  &    20 & 3 & r & \hc{4.37}{11} & $39.75 $ & $38.02 $ & $41.78 $\\ 
  &    &  & s  & $34.26 $ & $41.06 $ & $42.70 $ & $43.28 $\\ 
    \end{tabular}
  } & 
    {\small
      \begin{tabular}{lc@{\hspace{2pt}}c@{\hspace{2pt}}c|rrrr}
        \multicolumn{8}{c}{\bf AllAml, KL}\\
   (i) &     $k_1$ & $k_2$ & & \cc{x} & \cc{xk} & \cc{xc} & \cc{xkc}\\ \cline{2-8}
     &   5 & 3 & r & \hc{5.64}{-1} & $43.92 $ & $47.14 $ & $46.73 $\\ 
     &   &  & s  & $33.39 $ & $43.12 $ & $46.68 $ & $46.44 $\\ \cline{2-8} 
     &   10 & 3 & r & \hc{4.67}{-1} & $40.11 $ & $41.72 $ & $42.57 $\\ 
     &   &  & s  & $31.04 $ & $39.87 $ & $42.78 $ & $42.75 $\\ \cline{2-8} 
     &   20 & 3 & r & \hc{3.78}{-1} & $29.29 $ & $32.67 $ & $33.24 $\\ 
     &   &  & s  & $20.58 $ & $29.74 $ & $33.90 $ & $34.07 $\\ 
      \end{tabular}
    }\\
    \phantom{blank} & \phantom{line}\\
    {\small
    \begin{tabular}[htbp]{@{\hspace{2pt}}l@{\hspace{2pt}}c@{\hspace{2pt}}c@{\hspace{2pt}}|@{\hspace{4pt}}r@{\hspace{8pt}}r@{\hspace{8pt}}r@{\hspace{8pt}}r}
      (ii) &   $k_1$ & $k_2$ & \cc{rc} & \cc{rkc} & \cc{sc} & \cc{skc}\\ \cline{2-7}
     & 5 & 3& $ 13.8±3.7$& $ 2.8±1.2$& $ 5.2±1.7$& $ 3.0±1.4$\\
     & 10 & 3& $ 15.9±4.6$& $ 3.4±1.3$& $ 4.8±1.2$& $ 2.9±1.0$\\
     & 20 & 3& $ 12.3±3.4$& $ 2.9±1.3$& $ 4.7±1.3$& $ 2.9±0.9$\\
    \end{tabular} } & 
    {\small
      \begin{tabular}{@{\hspace{2pt}}l@{\hspace{2pt}}c@{\hspace{2pt}}c@{\hspace{2pt}}|@{\hspace{4pt}}r@{\hspace{8pt}}r@{\hspace{8pt}}r@{\hspace{8pt}}r}
      (ii) &   $k_1$ & $k_2$ & \cc{rc} & \cc{rkc} & \cc{sc} & \cc{skc}\\ \cline{2-7}
      &  5 & 3& $ 17.9±3.5$& $ 7.0±3.5$& $ 11.7±5.0$& $ 7.8±4.1$\\
      &  10 & 3& $ 18.3±3.4$& $ 7.2±2.6$& $ 12.1±3.5$& $ 9.3±4.6$\\
      &  20 & 3& $ 18.9±2.5$& $ 12.0±4.5$& $ 11.1±3.1$& $ 10.3±2.9$\\
      \end{tabular}
    }
  \end{tabular}

\vskip 10pt

  \begin{tabular}{@{\hspace{-2pt}}l@{\hspace{1pt}}l}
    {\small
      \begin{tabular}{lc@{\hspace{2pt}}c@{\hspace{2pt}}c|rrrr}
      \multicolumn{8}{c}{\bf Leukemia, Euc.}\\
   (i) &     $k_1$ & $k_2$ & & \cc{x} & \cc{xk} & \cc{xc} & \cc{xkc}\\
   \cline{2-8}
   &   3 & 2 & r & \hc{7.61}{4} & $5.48 $ & $5.77 $ & $6.74 $\\ 
   &   &  & s  & $0.17 $ & $5.54 $ & $5.73 $ & $6.78 $\\ \cline{2-8} 
   &   3 & 3 & r & \hc{7.57}{4} & $6.53 $ & $7.18 $ & $7.75 $\\ 
   &   &  & s  & $0.14 $ & $6.79 $ & $6.77 $ & $7.79 $\\ \cline{2-8} 
   &   50 & 2 & r & \hc{7.30}{4} & $3.79 $ & $5.97 $ & $7.25 $\\ 
   &   &  & s  & $0.33 $ & $3.75 $ & $5.54 $ & $7.25 $\\ \cline{2-8} 
   &   50 & 3 & r & \hc{7.15}{4} & $4.90 $ & $7.34 $ & $8.93 $\\ 
   &   &  & s  & $0.60 $ & $5.00 $ & $8.00 $ & $9.06 $\\ \cline{2-8} 
   &   75 & 2 & r & \hc{7.26}{04} & $3.66 $ & $5.67 $ & $6.89 $\\ 
   &   &  & s  & $0.02 $ & $3.67 $ & $5.23 $ & $6.88 $\\ \cline{2-8} 
   &   75 & 3 & r & \hc{7.09}{4} & $4.59 $ & $7.09 $ & $8.47 $\\ 
   &   &  & s  & $0.60 $ & $4.61 $ & $7.05 $ & $8.52 $\\ 
    \end{tabular}
  } & 
  {\small
     \begin{tabular}{lc@{\hspace{2pt}}c@{\hspace{2pt}}c|rrrr}
       \multicolumn{8}{c}{\bf Leukemia, KL}\\
   (i) &     $k_1$ & $k_2$ & & \cc{x} & \cc{xk} & \cc{xc} & \cc{xkc}\\ \cline{2-8}
     &  3 & 2 & r & \hc{1.82}{-1} & $5.11 $ & $7.15 $ & $7.52 $\\ 
     &  &  & s  & $0.36 $ & $4.93 $ & $7.19 $ & $7.51 $\\ \cline{2-8} 
     &  3 & 3 & r & \hc{1.81}{-1} & $6.00 $ & $8.13 $ & $8.76 $\\ 
     &  &  & s  & $0.44 $ & $6.08 $ & $8.18 $ & $8.76 $\\ \cline{2-8} 
     &  50 & 2 & r & \hc{1.71}{-1} & $3.81 $ & $7.58 $ & $7.60 $\\ 
     &  &  & s  & $-0.21 $ & $3.65 $ & $7.32 $ & $7.35 $\\ \cline{2-8} 
     &  50 & 3 & r & \hc{1.68}{-1} & $4.74 $ & $9.31 $ & $9.35 $\\ 
     &  &  & s  & $1.08 $ & $5.16 $ & $9.70 $ & $9.75 $\\ \cline{2-8} 
     &  75 & 2 & r & \hc{1.71}{-1} & $3.36 $ & $6.92 $ & $6.95 $\\ 
     &  &  & s  & $-0.35 $ & $2.85 $ & $6.60 $ & $6.30 $\\ \cline{2-8} 
     &  75 & 3 & r & \hc{1.66}{-1} & $4.48 $ & $9.04 $ & $9.11 $\\ 
     &  &  & s  & $0.69 $ & $4.25 $ & $8.66 $ & $8.68 $\\ 
     \end{tabular}
   }\\
   \phantom{blank} & \phantom{line}\\
    {\small
      \begin{tabular}{@{\hspace{2pt}}l@{\hspace{2pt}}c@{\hspace{2pt}}c@{\hspace{2pt}}|@{\hspace{4pt}}r@{\hspace{8pt}}r@{\hspace{8pt}}r@{\hspace{8pt}}r}
      (ii) &   $k_1$ & $k_2$ & \cc{rc} & \cc{rkc} & \c{sc} & \cc{skc}\\ \cline{2-7}
   &   3 & 2& $ 3.8±1.3$& $ 2.0±0.0$& $ 3.3±0.8$& $ 2.0±0.0$\\
   &   3 & 3& $ 4.5±1.5$& $ 2.2±0.4$& $ 3.8±1.1$& $ 2.1±0.3$\\
   &   50 & 2& $ 3.3±1.1$& $ 2.0±0.0$& $ 2.9±1.3$& $ 2.0±0.0$\\
   &   50 & 3& $ 3.3±0.8$& $ 2.0±0.0$& $ 3.7±1.1$& $ 2.0±0.0$\\
   &   75 & 2& $ 3.1±0.9$& $ 2.0±0.0$& $ 3.3±1.1$& $ 2.0±0.0$\\
   &   75 & 3& $ 3.6±0.9$& $ 2.0±0.0$& $ 3.4±1.0$& $ 2.0±0.0$\\
     \end{tabular}
   } & 
   {\small
     \begin{tabular}{@{\hspace{2pt}}l@{\hspace{2pt}}c@{\hspace{2pt}}c@{\hspace{2pt}}|@{\hspace{4pt}}r@{\hspace{8pt}}r@{\hspace{8pt}}r@{\hspace{8pt}}r}
      (ii) &   $k_1$ & $k_2$ & \cc{rc} & \cc{rkc} & \cc{sc} & \cc{skc}\\ \cline{2-7}
      & 3 & 2& $ 7.6±3.5$& $ 4.5±3.2$& $ 8.0±2.9$& $ 4.6±3.2$\\
      & 3 & 3& $ 7.4±2.5$& $ 5.1±1.7$& $ 7.3±3.0$& $ 4.7±1.4$\\
      & 50 & 2& $ 5.4±1.8$& $ 3.4±0.7$& $ 5.7±2.5$& $ 3.3±0.5$\\
      & 50 & 3& $ 6.2±2.0$& $ 4.5±0.8$& $ 5.5±1.0$& $ 4.6±1.1$\\
      & 75 & 2& $ 5.3±1.8$& $ 3.4±1.2$& $ 5.6±2.2$& $ 3.2±0.5$\\
      & 75 & 3& $ 5.6±1.4$& $ 4.2±0.6$& $ 4.9±1.1$& $ 4.1±0.3$\\
     \end{tabular}
   }
 \end{tabular}
  \label{tab:leukemia}
\end{table}

\begin{table}[htbp]
\caption{\small (i) Improvement of CoTeC and SiTeC variants upon
    `\textbf{r}' in \%; the respective reference value ($J_2$ for
    `\textbf{r}') is shaded in gray. 
    (ii) Average number of SiTeC iterations. } 
  \vspace{1.5mm}

    \centering
    {\small
      \begin{tabular}{@{\hspace{2pt}}l@{\hspace{2pt}}c@{\hspace{2pt}}c@{\hspace{2pt}}c|rrrr}
        \multicolumn{8}{c}{\bf Mll, Euc.}\\
   (i) &     $k_1$ & $k_2$ & & \cc{x} & \cc{xk} & \cc{xc} & \cc{xkc}\\ \cline{2-8}
       & 3 & 3 & r & \hc{6.52}{4} & $10.54 $ & $11.26 $ & $11.45 $\\ 
       & &  & s  & $1.41 $ & $10.62 $ & $11.20 $ & $11.46 $\\ \cline{2-8} 
      &  50 & 3 & r & \hc{5.83}{4} & $8.40 $ & $12.53 $ & $13.21 $\\ 
      &  &  & s  & $1.12 $ & $8.23 $ & $12.35 $ & $13.17 $\\ \cline{2-8} 
      &  75 & 3 & r & \hc{5.75}{4} & $7.84 $ & $11.69 $ & $12.52 $\\ 
      &  &  & s  & $0.84 $ & $7.86 $ & $11.68 $ & $12.52 $\\
      \end{tabular}
      \begin{tabular}[htbp]{l@{\hspace{2pt}}c@{\hspace{2pt}}c@{\hspace{2pt}}|@{\hspace{4pt}}r@{\hspace{8pt}}r@{\hspace{8pt}}r@{\hspace{8pt}}r}
      (ii) &   $k_1$ & $k_2$ & \cc{rc} & \cc{rkc} & \cc{sc} & \cc{skc}\\ \cline{2-7}
      &  3 & 3& $ 4.2±1.2$& $ 2.0±0.3$& $ 3.8±1.0$& $ 2.0±0.5$\\
      &  50 & 3& $ 4.7±1.9$& $ 2.0±0.0$& $ 4.2±1.5$& $ 2.1±0.3$\\
      &  75 & 3& $ 4.4±1.4$& $ 2.0±0.0$& $ 4.3±1.5$& $ 2.0±0.0$\\
    \end{tabular}
  }
  \label{tab:mll}
\end{table}

\section{Conclusions}
In this paper we presented a simple, and to our knowledge the first
approximation algorithm for Bregman and metric tensor clustering. Our
approximation factor grows linearly with the 
order $m$ of the tensor for Bregman divergences, and is slightly
superlinear in $m$ for arbitrary metrics. It is always linear in the
quality of the sub-clusterings.

Our experiments demonstrated the dependence of the multi-dimensional
clustering on the single-dimensional clusterings, confirming the dependence
stated in the theoretical bound. On real-world data, the approximation
algorithm is also suitable as an initialization for a simultaneous
co-clustering algorithm, and endows the latter with its approximation
guarantees. In fact the approximation algorithm by itself can also yield
reasonable results in practice.

In our experiments we used single-dimensional clusterings with guarantees
for our overall approximation algorithm. An interesting direction for future
work is the development of a simultaneous approximation algorithm, such as a
specific co-clustering seeding scheme of multi-dimensional centers, which
can be then used as a subroutine by our tensor clustering algorithm.

\bibliographystyle{plainnat}
\bibliography{paper}

\begin{thebibliography}{36}
\providecommand{\natexlab}[1]{#1}
\providecommand{\url}[1]{\texttt{#1}}
\expandafter\ifx\csname urlstyle\endcsname\relax
  \providecommand{\doi}[1]{doi: #1}\else
  \providecommand{\doi}{doi: \begingroup \urlstyle{rm}\Url}\fi

\bibitem[Ackermann and Bl\"omer(2009)]{acker2}
M.~R. Ackermann and Johannes Bl\"omer.
\newblock {Coresets and Approximate Clustering for Bregman Divergences}.
\newblock In \emph{Proc. 20th ACM-SIAM Symposium on Discrete Algorithms (SODA
  '09)}, 2009.
\newblock To appear.

\bibitem[Ackermann et~al.(2008)Ackermann, Blomer, and Sohler]{ackerman}
M.~R. Ackermann, J.~Blomer, and C.~Sohler.
\newblock Clustering for metric and non-metric distance measures.
\newblock In \emph{ACM-SIAM SODA}, April 2008.

\bibitem[Agarwal et~al.(2005)Agarwal, Lim, Zelnik-Manor, Perona, Kriegman, and
  Belongie]{tensor3}
S.~Agarwal, J.~Lim, L.~Zelnik-Manor, P.~Perona, D.~Kriegman, and S.~Belongie.
\newblock Beyond pairwise clustering.
\newblock In \emph{IEEE CVPR}, 2005.

\bibitem[Anagnostopoulos et~al.(2008)Anagnostopoulos, Dasgupta, and
  Kumar]{coclus}
A.~Anagnostopoulos, A.~Dasgupta, and R.~Kumar.
\newblock Approximation algorithms for co-clustering.
\newblock In \emph{PODS}, 2008.

\bibitem[Armstrong(2002)]{hyukref48}
S.~A. Armstrong.
\newblock Mll translocations specify a distinct gene expression profile that
  distinguishes a unique leukemia.
\newblock \emph{Nature Genetics}, 30:\penalty0 41--17, 2002.

\bibitem[Arthur and Vassilvitskii(2007)]{dasv07}
D.~Arthur and S.~Vassilvitskii.
\newblock {\texttt{k-means++}: The Advantages of Careful Seeding}.
\newblock In \emph{SODA}, pages 1027--1035, 2007.

\bibitem[Banerjee et~al.(2005)Banerjee, Merugu, Dhillon, and Ghosh]{banerjee}
A.~Banerjee, S.~Merugu, I.~S. Dhillon, and J.~Ghosh.
\newblock {Clustering with Bregman Divergences}.
\newblock \emph{JMLR}, 6\penalty0 (6):\penalty0 1705--1749, October 2005.

\bibitem[Banerjee et~al.(2007{\natexlab{a}})Banerjee, Basu, and
  Merugu]{arinmulti}
A.~Banerjee, S.~Basu, and S.~Merugu.
\newblock {Multi-way Clustering on Relation Graphs}.
\newblock In \emph{SIAM Data Mining}, 2007{\natexlab{a}}.

\bibitem[Banerjee et~al.(2007{\natexlab{b}})Banerjee, Dhillon, Ghosh, Merugu,
  and Modha]{coclust}
A.~Banerjee, I.~S. Dhillon, J.~Ghosh, S.~Merugu, and D.~S. Modha.
\newblock {A Generalized Maximum Entropy Approach to Bregman Co-clustering and
  Matrix Approximation}.
\newblock \emph{JMLR}, 8:\penalty0 1919--1986, 2007{\natexlab{b}}.

\bibitem[Bekkerman et~al.(2005)Bekkerman, El-Yaniv, and McCallum]{tensor2}
R.~Bekkerman, R.~El-Yaniv, and A.~McCallum.
\newblock Multi-way distributional clustering via pairwise interactions.
\newblock In \emph{ICML}, 2005.

\bibitem[Berg et~al.(1984)Berg, Christensen, and Ressel]{becr84}
C.~Berg, J.~Christensen, and P.~Ressel.
\newblock \emph{Harmonic Analysis on Semigroups: Theory of Positive Definite
  and Related Functions}.
\newblock Springer-Verlag, 1984.

\bibitem[Bregman(1967)]{breg67}
L.~M. Bregman.
\newblock The relaxation method of finding the common point of convex sets and
  its applications to the solution of problems in convex programming.
\newblock \emph{{U.S.S.R. Computational Mathematics and Mathematical Physics}},
  7\penalty0 (3):\penalty0 200--217, 1967.

\bibitem[Censor and Zenios(1997)]{censor:zenios}
Y.~Censor and S.~A. Zenios.
\newblock \emph{{Parallel Optimization: Theory, Algorithms, and Applications}}.
\newblock Oxford University Press, 1997.

\bibitem[Chaudhuri and McGregor(2008)]{mcgregor}
K.~Chaudhuri and A.~McGregor.
\newblock Finding metric structure in information theoretic clustering.
\newblock In \emph{Conf. on Learning Theory, COLT}, July 2008.

\bibitem[Cho and Dhillon(2008)]{cho08}
H.~Cho and I.~Dhillon.
\newblock Coclustering of human cancer microarrays using minimum sum-squared
  residue coclustering.
\newblock \emph{IEEE/ACM Transactions on Computational Biology and
  Bioinformatics}, 5\penalty0 (3):\penalty0 385--400, 2008.

\bibitem[Cho et~al.(2004)Cho, Dhillon, Guan, and Sra]{suv.coclus}
H.~Cho, I.~S. Dhillon, Y.~Guan, and S.~Sra.
\newblock {Minimum Sum Squared Residue based Co-clustering of Gene Expression
  data}.
\newblock In \emph{Proc.~4th SIAM International Conference on Data Mining
  (SDM)}, pages 114--125, Florida, 2004. SIAM.

\bibitem[de~Silva and Lim(2008)]{leklim}
V.~de~Silva and L.-H. Lim.
\newblock {Tensor Rank and the Ill-Posedness of the Best Low-Rank Approximation
  Problem}.
\newblock \emph{SIAM J. on Matrix Analysis and Applications}, 30\penalty0
  (3):\penalty0 1084--1127, 2008.

\bibitem[Dhillon et~al.(2003)Dhillon, Mallela, and Modha]{itcc}
I.~S. Dhillon, S.~Mallela, and D.~S. Modha.
\newblock Information-theoretic co-clustering.
\newblock In \emph{Proc. ACM SIGKDD 2003}, pages 89--98, 2003.

\bibitem[Drineas et~al.(2004)Drineas, Frieze, Kannan, Vempala, and
  Vinay]{drineas}
P.~Drineas, A.~Frieze, R.~Kannan, S.~Vempala, and V.~Vinay.
\newblock Clustering large graphs via the singular value decomposition.
\newblock \emph{Machine Learning}, 56:\penalty0 9--33, 2004.

\bibitem[Golub et~al.(1999)Golub, Slonim, Tamayo, Huard, Gaasenbeek, Mesirov,
  Coller, Loh, Downing, Caliguri, Bloomfield, and Lander]{hyukref9}
T.~R. Golub, D.~K. Slonim, P.~Tamayo, C.~Huard, M.~Gaasenbeek, J.~P. Mesirov,
  H.~Coller, M.L. Loh, J.~R. Downing, M.~A. Caliguri, C.~D> Bloomfield, and
  E.~S. Lander.
\newblock Molecular classification of cancer: Class discovery and class
  prediction by gene expression monitoring.
\newblock \emph{Science}, 286:\penalty0 531--537, 1999.

\bibitem[Govindu(2005)]{tensor4}
V.~M. Govindu.
\newblock A tensor decomposition for geometric grouping and segmentation.
\newblock In \emph{IEEE CVPR}, 2005.

\bibitem[Greub(1967)]{greub}
W.~H. Greub.
\newblock \emph{Multilinear {A}lgebra}.
\newblock Springer, 1967.

\bibitem[Hartigan(1975)]{hartigan}
J.~A. Hartigan.
\newblock \emph{Clustering {A}lgorithms}.
\newblock Wiley, 1975.

\bibitem[Hein and Bosquet(2005)]{hebo05}
M.~Hein and O.~Bosquet.
\newblock Hilbertian metrics and positive definite kernels on probability
  measures.
\newblock In \emph{AISTATS}, 2005.

\bibitem[Kluger et~al.(2003)Kluger, Basri, and Chang]{kluger03}
Y.~Kluger, R.~Basri, and J.~T. Chang.
\newblock Spectral biclustering of microarray data: Coclustering genes and
  conditions.
\newblock \emph{Genome Research}, 13:\penalty0 703--716, 2003.

\bibitem[Kolda and Sun(2008)]{jimeng}
T.~G. Kolda and J.~Sun.
\newblock {Scalable Tensor Decompositions for Multi-aspect Data Mining}.
\newblock In \emph{ICDM}, 2008.

\bibitem[Kumar et~al.(2004)Kumar, Sabharwal, and Sen]{kumar}
A.~Kumar, Y.~Sabharwal, and S.~Sen.
\newblock A simple linear time $(1+\epsilon)$-approximation algorithms for
  k-means clustering in any dimensions.
\newblock In \emph{IEEE Symp. on Foundations of Comp. Sci.}, 2004.

\bibitem[LLoyd(1982)]{lloyd}
S.~P. LLoyd.
\newblock Least squares quantization in {PCM}.
\newblock \emph{IEEE Tran. on Inf. Theory}, 28\penalty0 (2):\penalty0 129--136,
  1982.

\bibitem[Long et~al.(2006)Long, Wu, and Zhang]{tensor1}
B.~Long, X.~Wu, and Z.~Zhang.
\newblock Unsupervised learning on k-partite graphs.
\newblock In \emph{SIGKDD}, 2006.

\bibitem[Nock et~al.(2008)Nock, Luosto, and Kivinen]{ecmlbreg}
R.~Nock, P.~Luosto, and J.~Kivinen.
\newblock Mixed bregman clustering with approximation guarantees.
\newblock In \emph{Euro. Conf. on Mach. Learning (ECML)}, LNAI 5212, 2008.

\bibitem[Puolam\"{a}ki et~al.(2008)Puolam\"{a}ki, Hanhij\"{a}rvi, and
  Garriga]{coclus1}
K.~Puolam\"{a}ki, S.~Hanhij\"{a}rvi, and G.~C. Garriga.
\newblock An approximation ratio for biclustering.
\newblock \emph{Inf. Process. Lett.}, 108\penalty0 (2):\penalty0 45--49, 2008.

\bibitem[Schoenberg(1938)]{scho38}
I.~J. Schoenberg.
\newblock Metric spaces and positive definite functions.
\newblock \emph{Transactions of American Mathematical Society}, 44\penalty0
  (3):\penalty0 522--536, 1938.

\bibitem[Sch{\"o}lkopf and Smola(2001)]{scsm01}
B.~Sch{\"o}lkopf and A.~Smola.
\newblock \emph{Learning with Kernels}.
\newblock MIT Press, 2001.

\bibitem[Shashua et~al.(2006)Shashua, Zass, and Hazan]{shashua}
A.~Shashua, R.~Zass, and T.~Hazan.
\newblock {Multi-way Clustering Using Super-Symmetric Non-negative Tensor
  Factorization}.
\newblock \emph{LNCS}, 3954:\penalty0 595--608, 2006.

\bibitem[Sra et~al.(2008)Sra, Jegelka, and Banerjee]{kmeanstr}
S.~Sra, S.~Jegelka, and A.~Banerjee.
\newblock Approximation algorithms for bregman clustering co-clustering and
  tensor clustering.
\newblock Technical Report 177, MPI for Biological Cybernetics, Oct. 2008.

\bibitem[Zha et~al.(2008)Zha, Ding, Li, and Zhu]{kddwkshp}
H.~Zha, C.~Ding, T.~Li, and S.~Zhu.
\newblock {Workshop on Data Mining using Matrices and Tensors}.
\newblock KDD, 2008.

\end{thebibliography}

\end{document}